\documentclass[journal,comsoc]{IEEEtran}
\usepackage[T1]{fontenc}
\usepackage{amsmath,amsthm}
\usepackage[cmintegrals]{newtxmath}
\usepackage[nocompress]{cite}
\usepackage[utf8]{inputenc}
\usepackage[english]{babel}
\usepackage[dvips]{graphicx}
\usepackage{subcaption}
\usepackage{algorithm,algorithmic}
\usepackage{caption}
\usepackage{bm}
\usepackage[noblocks]{authblk}
\usepackage{color}
\usepackage{multirow}
\usepackage{siunitx}
\usepackage{setspace}

\newtheorem{theorem}{Theorem}[section]

\floatname{Procedure}{algorithm}

\begin{document}
\thispagestyle{empty}
This work has been submitted to the IEEE for possible publication.  Copyright may be transferred without notice, after which this version may no longer be accessible.
\clearpage
\setcounter{page}{1}

\title{Optimizing IoT Energy Efficiency on Edge (EEE): a Cross-layer Design in a Cognitive Mesh Network}

\author{Jianqing Liu,
        Yawei Pang,
        Haichuan Ding,
        Ying Cai,
        Haixia Zhang,
        Yuguang Fang
\thanks{J. Liu is with the Department
of Electrical and Computer Engineering, University of Alabama in Huntsville, Huntsville,
AL 35899 USA e-mail: jianqing.liu@uah.edu}
\thanks{Y. Pang is with the Department of IoT, School of Computer and Software, Nanjing University of Information Science and Technology, Nanjing 210044, China e-mail: yaweipang@gmail.com}
\thanks{H. Ding is with the Department of Electrical Engineering and Computer Science, University of Michigan, Ann Arbor, MI 48109 USA e-mail: dhcbit@gmail.com}
\thanks{Y. Cai is with the Computer School, Beijing Information Science \&
Technology University, Beijing 100101, China e-mail: ycai@bistu.edu.cn}
\thanks{H. Zhang is with the School of Control Science and Engineering, Shandong University, Jinan 250100, China e-mail: haixia.zhang@sdu.edu.cn}
\thanks{Y. Fang is with the Department
of Electrical and Computer Engineering, University of Florida, Gainesville, FL 32611 USA e-mail: fang@ece.ufl.edu.}}
\maketitle

\begin{abstract}
Battery-powered wireless IoT devices are now widely seen in many critical applications. Given the limited battery capacity and inaccessibility to external power recharge, optimizing energy efficiency (EE) plays a vital role in prolonging the lifetime of these IoT devices. However, a sheer amount of existing works only focus on the EE design at the infrastructure level such as base stations (BSs) but with little attention to the EE design at the device level. In this paper, we propose a novel idea that aims to shift energy consumption to a grid-powered cognitive radio mesh network thus preserving energy of battery-powered devices. Under this line of thinking, we cast the design into a cross-layer optimization problem with an objective to maximize devices' energy efficiency. To solve this problem, we propose a parametric transformation technique to convert the original problem into a more tractable one. A baseline scheme is used to demonstrate the advantage of our design. We also carry out extensive simulations to exhibit the optimality of our proposed algorithms and the network performance under various settings.
\end{abstract}


\IEEEpeerreviewmaketitle

\section{Introduction}
Over the last few years, the explosive growth of smart devices is accelerating the advent of Internet-of-Things (IoT). Among these IoT devices, a significant number of them run on wireless radio and are powered by battery \cite{iot17wireless,iot19battery}. It is simply because they are deployed in hard-to-access locations, infeasible to be hard wired, or just designed to be user friendly. Such use cases sometimes require battery-powered IoT devices to operate for a significant amount of time without battery replacement, 10 years for instance \cite{lu2017reaching}. To achieve this goal, a few low-power wireless technologies such as Bluetooth Low Energy (BLE) \cite{siekkinen2012low}, ZigBee (IEEE 802.15.4e) \cite{juc2016energy}, WiFi (IEEE 802.11ah) \cite{akeela2017design}, low-power long-range (LoRa) \cite{casals2017modeling} are developed with the objective of optimizing a device's radio interface for higher energy efficiency (e.g., via better power control and transmission scheduling). However, such optimization with energy efficiency as the objective usually comes with the sacrifice to other metrics. For instance, long duty-cycle incurs high latency \cite{palattella2013optimal}; low transmit power reduces data rate \cite{reynders2017power} and simplified radio even hampers security \cite{ryan2013bluetooth}. In light of this, alternative solutions to increasing IoT devices' energy efficiency are worthwhile to investigate.

It has been well recognized that the attempt in increasing energy efficiency is constrained by the spectrum efficiency \cite{chen2011fundamental}. The reason lies in the Shannon's capacity theorem, which reveals that link capacity increases only logarithmically with power but linearly with bandwidth, indicating that bandwidth could more effectively bring down the power consumption. With the observation that a large portion of licensed spectrum is not well utilized in certain areas \cite{yin2012mining}, one can harvest and opportunistically access under-utilized spectrum. As an enabling technology, cognitive radio (CR) \cite{haykin2005cognitive} promises to realize dynamic spectrum access and thus increase spectrum efficiency.

In this paper, we propose a novel idea - by allowing battery-powered IoT devices to shift their energy consumption to grid-powered CR-capable devices, IoT devices' energy efficiency can be increased. This energy-efficiency-on-edge (EEE) idea is motivated by the observation that IoT devices are more sensitive to energy consumption than grid-powered devices are. In other words, it is worthwhile to spend a bit more grid energy on grid-powered devices so as to somewhat, if not significantly, improve the energy efficiency of battery-power devices. To enable such ``energy shift'', we leverage a cognitive capacity harvesting network (CCHN) which was firstly proposed in our previous works \cite{ding2017cognitive} as shown in Fig.\ref{sys_mod}. The essence of this architecture is that the routers, also called CR routers, have CR capabilities and they can form a multi-hop mesh network to get closer to and help the end devices with no CR capabilities for data exchange, potentially saving energy and increasing spectrum efficiency. Under this structure, we intend to augment it with a centralized IoT system\footnote{In this work, a hypothetical model instead of a realistic system is used to convey the idea.} (e.g., LoRAWAN, LTE) and let the CCHN help battery-powered IoT devices to relay data, thus increasing their energy efficiency.

Specifically, we consider the uplink transmission where IoT devices can send data either directly to the BS/gateway or to CR routers. For the former case, it is a one-hop transmission while in the latter one, an IoT device is firstly connected to a CR router using its existing radio interface (e.g., LoRa, WiFi) and its traffic are then delivered to the BS via multi-hop transmissions using harvested bands. The rationale of studying this problem is that IoT devices could use lower transmission power to associate with closer CR routers but may not be able to always receive satisfactory service through the CCHN due to the uncertainty of the harvested bands; while they could obtain reliable throughput via direct connection to the BS but may need to apply higher transmission power. Obviously, there is a tradeoff between service quality (i.e., reliability and throughput) and power consumption. Therefore, we will investigate the energy efficient design in this respect, with the design dimension to be device association, uplink power control and channel allocation, and multi-hop scheduling and routing in the CCHN.

Towards this design objective, there are several technical challenges to be addressed. First, the weighted sum-of-ratios form of the objective function is non-convex, which makes the optimization problem difficult to tackle. Second, the availability of harvested spectrum is highly unpredictable and the usable bandwidth in this regard should be naturally modeled as a random variable, which results in a stochastic constraint and makes the problem intractable. Third, the association variable is in integer form and tightly coupled with other decision variables so solving the problem via conventional approaches is highly prohibitive especially when the network size is large. In light of these aforementioned challenges, we propose corresponding solution to each of them, which in turn demonstrates our technical contributions as follows:
\begin{itemize}
  \item We introduce auxiliary variables and transform the original objective function into a parametric subtractive form which bears the desired convexity property. Their equivalence in terms of finding the same solution is further proved.
  \item We reformulate the stochastic constraint of the availability of harvested spectrum as a chance constraint of $\Delta$-confidence level. Then, the feasible region of original optimization problem becomes a convex set.
  \item We address the integer programming part through a two-step procedure: relaxing and then rounding. To decouple the decision variables, we further apply the dual-based approach to make the original problem more tractable.
\end{itemize}

The rest of the paper is organized as follows. Section II introduces the most recent literature of this topic. Section III describes the system model. The problem formulation is outlined in Section IV. We propose solution algorithms in Section V and present the performance evaluation in Section VI. Finally, Section VII concludes the paper.

\section{Related Works}
IoT wireless access technology can be broadly classified into random access (or contention-based) and deterministic access (or connection-based) categories \cite{centenaro2017comparison}. Exemplary standard of the former is LTE-RACH \cite{polese2016m2m} while the latter could be TDMA-based (e.g., BLE) and OFDM-based (e.g., 802.11ah). There has been a flux of works on optimizing these standards to make them energy-efficient, but the drawbacks are also evident in the sense that other performance metrics such as latency, throughput and security could be hampered \cite{palattella2013optimal,reynders2017power,ryan2013bluetooth}. Since our work is not constrained to any existing standard, we will survey generic EE design in hypothetical wireless networks. Nevertheless, the number of related works is still significant so we limit the scope to the OFDM-based access systems and future networks like cognitive radio networks (CRNs) and the CRN-enabled networks.

The well-recognized EE design model is to maximize energy efficiency (defined as the ratio of rate to power consumption) under wireless resource constraints \cite{zappone2015energy}. For instance, in an OFDM-based radio access network, there are many related works to optimize EE by jointly considering power control and subcarrier allocation \cite{cheung2013achieving,xiong2016energy,he2014leakage,zarakovitis2016maximizing}. Cheung \emph{et al.} \cite{cheung2013achieving} focused on a multi-relay assisted OFDM cellular network and studied the EE maximization problem by formulating EE as the ratio of total network throughput to total power consumption at a BS. Xiong \emph{et al.} also targeted at the same wireless setting but investigated the EE maximization constrained on users' proportional rate fairness \cite{xiong2016energy}. However, these papers and their related ones bear two major shortcomings: (i) they modeled the network EE in a ``sum-to-sum'' form (instead of the ``sum-of-ratios'' form) which fails to capture each device's EE; (ii) they used the infrastructure's power consumption in calculating EE, which lacks emphasis on the battery-powered end devices. Unfortunately, there is a lack of study when it comes to addressing these two issues. Recent works \cite{he2014leakage,zarakovitis2016maximizing} nonetheless have some merits along this line. Zarakovitis \emph{et al.} \cite{zarakovitis2016maximizing} studied the EE design in a downlink OFDM cellular system by characterizing EE in weighted sum-of-ratios. They applied the Maclaurin series expansion to transform the objective into a tractable form and then solved it in polynomial time. He \emph{et al.} \cite{he2014leakage} focused on a multi-cell downlink OFDM cellular system with coordinated beamforming. They introduced auxiliary variables to transform the weighted sum-of-ratios objective into a parametric subtractive form, which afterwards became easier to address. These works adopted the ``sum-of-ratios'' to define the EE objective but still employed the infrastructure's rather than the device's power consumption.

The EE design in CRNs and CRN-enabled networks is still in its infancy. Amongst the limited works, Wang \emph{et al.} in \cite{wang2015energy} considered an OFDM-based CRNs and investigated the EE design by taking channel uncertainty into consideration. Xie \emph{et al.} \cite{xie2012energy} proposed a new cognitive cellular network architecture consisting of macrocells and femtocells. They utilized game theoretic approaches to investigate the energy efficient resource allocation. There are other similar works in this context \cite{bayhan2013scheduling,bedeer2015energy}, but they all bear two critical limitations: (i) end devices are assumed to have CR capabilities which in nature is not energy efficient due to the tedious spectrum sensing; (ii) the EE measurement is the ratio of overall rate to overall power consumption, which is in the ``sum-to-sum'' form. To cope with the first challenge, a  cognitive capacity harvesting network (CCHN) was proposed in \cite{ding2017cognitive} and its basics are discussed in Section II. Under this architecture, Ding \emph{et al.} \cite{ding2018session,ding2017smart} and Liu \emph{et al.} \cite{liu2016energy,liu2015energy} developed protocols to achieve higher throughput and energy efficiency for lightweight end devices.

In this work, we exploit the CCHN-enabled OFDM access system to maximize energy efficiency of battery-powered IoT devices. We characterize the network EE using the summation of end devices' EE, and model the problem under a novel idea - ``energy shift'' from battery-powered devices to grid-powered ones.
\section{System Model}
\subsection{Network Description}
In this paper, we consider a CCHN-augmented network which consists of a secondary service provider (SSP), a BS, multiple CR routers and end users\footnote{We use ``end users'' and ``IoT devices'' interchangeably in the following article to represent the same meaning.}, as shown in Fig.\ref{sys_mod}. Specifically, SSP is a wireless service provider which has its own licensed bands, typically called the ``basic bands'', for reliable control signalling, handling handovers and so on. SSP could also harvest spectrum bands from other operators via paradigms such as spectrum sensing or spectrum auction. As the centralized coordinator, SSP observes and collects network information (e.g., users' traffic demands, channel state information) in its coverage area and then performs network optimization (e.g., power control, channel allocation, link scheduling and routing) to determine the optimal approaches for service provisioning. CR routers are grid-powered devices with CR capabilities, and they form a mesh network that is capable of using the harvested bands to transmit data. BS has multiple radio interfaces and serves as the gateway to the Internet for CR routers. In this architecture, end users do not have to possess CR capabilities. CR routers could tune their radio interfaces to what end users  use to make connections. Due to the close proximity between CR routers and end users, the frequency reuse ratio and devices' energy efficiency are greatly enhanced. For more details of this architecture, interested readers are referred to \cite{ding2017cognitive}.
\begin{figure}[!htb]
  \begin{center}
  \includegraphics[width=3.2in]{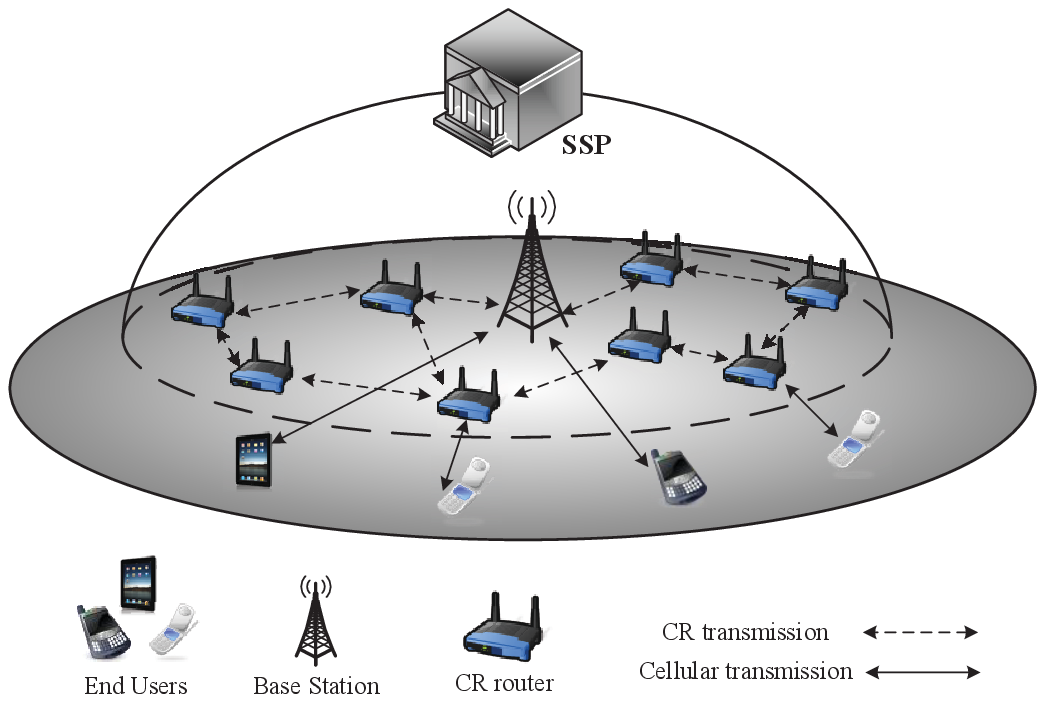}
  \end{center}
  \begin{center}
   \parbox{8cm}{\caption{The CCHN-augmented IoT Network.}\label{sys_mod}}
  \end{center}
\end{figure}

\subsection{Network model}
In this paper, we focus on uplink data transmission from IoT devices to a BS. As shown in Fig.\ref{sys_mod}, suppose a set of battery-powered devices/users, say, ${\cal U} = \{ 1,...,u,...,U\}$, each of which initializes a session whose destination is the BS denoted by $b$. We index the set of sessions as ${\cal L} = \{ {l_1},...,{l_u},...,{l_U}\}$ and let $s({l_u})$ and $d({l_u})$ denote the source (i.e., $s({l_u})=u$) and destination (i.e., $d({l_u})=b$) of session ${l_u} \in {\cal L}$, respectively. Also consider the network consisting of $K$ CR routers $\underline {\cal K}  = \{ 1,...,k,...,K\}$ and together with the BS $b$, we denote ${\cal K} = \underline {\cal K}  \cup \{ b\}$ as the set of grid-powered wireless infrastructures. Suppose the network applies single carrier-frequency division multiplexing access (SC-FDMA) for end users' uplink transmissions, where the network's basic band is divided into ${N_{tot}}$ number of orthogonal sub-channels which are shared among end users. Denote the harvested band as $m$ and it is allocated to the mesh network to form multi-hop transmissions.

Our model is applicable to a low mobility scenario. In such an environment, resource allocation can be conducted during the channel coherence time when channel is regarded as static. In light of this, we can just apply the line-of-sight (LOS) channel model instead of fast fading ones. Moreover, although signaling and computation overhead is incurred when solving the cross-layer optimization problem, the solution is applicable over a relatively long time scale due to users' low mobility and small variation of network parameters, which makes such overhead tolerable in the long run.

\section{Problem Formulation}
Given the system model described before, we target at the problem of uplink end-to-end data delivery from end users to the BS, where the user association, uplink power control, channel allocation, routing and link scheduling are jointly considered so as to maximize the network wide end users' energy efficiency.
\subsection{Uplink SC-FDMA}
In practical systems, similar to specifications in the 3GPP LTE standard, the uplink SC-FDMA in our design also implies restrictions on power and channel allocation \cite{myung2006single}. First, any sub-channel can only be assigned to one user, called \emph{exclusiveness}. Second, every user's allocated sub-channels must be continuous, called \emph{adjacency}. Third, user's transmit power should be identical on any allocated sub-channel in order to retain a low peak-to-average power ratio (PAPR).

Denote the indicator ${x_{u,k}}$ as the association variable, where ${x_{u,k}} = 1$ implies that user $u$ is associated with infrastructural node $k$, and ${x_{u,k}} = 0$ otherwise. Normally, the association is assumed to be performed in a large scale compared to the variation of channel so the fast fading is averaged out over the association time \cite{ye2013user}. We also consider a relatively low mobility environment, where the resource allocation is carried out during the channel coherence time so the channel can be regarded as static. Furthermore, in our network, to simplify the problem, we assume the ${N_{tot}}$ sub-channels are not reused so as to avoid strong interference.

\subsubsection{User Association}
Any end user can only be associated with the BS or a CR router, but not both. This physical constraint can be expressed as follows
\begin{align}
&\sum\limits_{k \in {\cal K}} {{x_{u,k}}} = 1, &\forall u \in {\cal U}, \label{ass_sum} \\
&{x_{u,k}} \in \{ 0,1\}, &\forall u \in {\cal U}, &\forall k \in {\cal K}. \label{ass_one}
\end{align}

\subsubsection{Sub-channel Allocation}
Suppose the end user $u$ is allocated with ${N_u}$ number of sub-channels which are contiguous in nature\footnote{Note that our work can be naturally extended to incorporate the case of finding the optimal sub-channel allocation pattern (i.e., a specific chunk of ${N_u}$ sub-channel collections) \cite{wong2009optimal,aijaz2014energy}, and we will leave it for the future work.}. Due to the property of \emph{exclusiveness}, the total allocated sub-channels cannot exceed ${N_{tot}}$, which is stated in below
\begin{equation}\label{cha_sum}
\sum\limits_{u \in {\cal U}} {{N_u}}  \le {N_{tot}}, {N_u} \in {\mathbb{Z}^ +}
\end{equation}
where ${\mathbb{Z}^ +}$ denotes the set of nonnegative integers.

\subsubsection{Link Capacity}
Given the sub-channel allocation, according to Shannon-Hartley theorem, the capacity of the link between an end user $u$ and the infrastructural node $j$ can be calculated as follows
\begin{equation}\label{link_cap}
{c_{u,j}} = {N_u}W\cdot{\log _2}(1 + \frac{{{p_u}|{h_{u,j}}{|^2}}}{{{N_u}W\cdot{N_0}}}), \forall u \in {\cal U}, \forall j \in {\cal K},
\end{equation}
where $W$ is the bandwidth of each sub-channel, ${N_0}$ is the power spectrum density (of unit $W/Hz$) of the Additive White Gaussian Noise (AWGN), and ${h_{u,k}}$ is the channel fading coefficient. Here, ${p_u}$ is denoted as the transmit power of the end user $u$ and as we know, ${p_u}$ must be evenly distributed across the allocated sub-channels to retain a low PAPR, which is given by $\frac{{{p_u}}}{{{N_u}}}$. Therefore, the link capacity is calculated as the sum of all allocated sub-channels' capacities, as shown in (\ref{link_cap}).

\subsubsection{Power Control}
Due to the hardware constraint, the transmit power of an end user cannot exceed its maximum allowable power level ${P_{u, \max}}$. Since end users have different power capabilities, ${P_{u, \max}}$ is a user-dependent variable. Moreover, we do not consider the power control of CR routers so their transmit powers are assumed to be fixed.

\subsection{The Cognitive Mesh Network}
When an end user is associated with a CR router, its traffic is delivered to the BS via multi-hop transmissions in the cognitive mesh network. In other words, the SSP allocates the harvested band $m$ to the mesh network and performs link scheduling and routing optimization to determine how to assist the end user $u$ to complete its session ${l_u}$ whose destination is the BS $b$. In what follows, we investigate the link scheduling and routing problem in the cognitive mesh network.

\subsubsection{Transmission Range and Interference Range}
Following the widely used model \cite{ding2018session}, we define the power propagation gain from the CR router $i$ ($\forall i \in \underline {\cal K}$) to another infrastructural node (either a CR router or the BS) $j$ ($\forall j \in {\cal K}\backslash i$) as ${g_{i,j}} = \zeta  \cdot d_{i,j}^{ - \gamma }$, where $\zeta$ is the antenna gain, ${d_{i,j}}$ refers to the Euclidean distance between $i$ and $j$, and $\gamma$ is the path loss exponent. Let assume that CR routers apply the same constant transmit power ${P_t}$ and define that the transmission is successful only when the received signal power exceeds a threshold $P_r^{th}$, i.e., ${P_t} \cdot {g_{i,j}} \ge P_r^{th}$. Then, we can obtain the transmission range of CR router $i$ as $R_i^T = {({P_t} \cdot \zeta /P_r^{th})^{1/\gamma }}$. Accordingly, we define the set of infrastructural nodes being in the transmission range of the CR router $i$ ($\forall i \in \underline {\cal K}$) as
\begin{equation}\label{tx_nb}
{{\cal T}_i} = \{ j \in {\cal K}|{d_{i,j}} \le R_i^T,j \ne i\}.
\end{equation}

On the other hand, to efficiently use harvested bands, the SSP should ensure the transmissions over different links do not conflict with each other. In light of this, we define the interference range in a similar way as before. Suppose the received interference can be ignored only when the received power is less than a threshold $P_I^{th}$, i.e., ${P_t} \cdot {g_{i,j}} < P_I^{th}$. Therefore, the interference range of the CR router $i$ ($\forall i \in \underline {\cal K}$) can be obtained as $R_i^I = {({P_t} \cdot \zeta /P_I^{th})^{1/\gamma }}$. Accordingly, the set of infrastructural nodes being in the interference range of the CR router $i$ ($\forall i \in \underline {\cal K}$) is defined as
\begin{equation}\label{int_nb}
{{\cal I}_i} = \{ j \in {\cal K}|{d_{i,j}} \le R_i^I,j \ne i\}.
\end{equation}

\subsubsection{Conflict Graph and Independent Sets \cite{ding2018session}}
Given the prior definition of the interference range, we can claim that two communication links conflict if the receiver of one link is within the interference range of the transmitter of the other link. A conflict graph $G = (V,E)$ is used to characterize the interfering relationship among different infrastructural links. Specifically, each vertex $v$ indicates a transmission link $(i,j)$ ($\forall i \in \underline {\cal K}$, $\forall j \in {\cal K}\backslash i$) and two links are said to be conflicted if there is an edge $e$ connecting the two corresponding vertices.

With this conflict graph being created, we can define an independent set (IS), which consists of a set of vertices $I \subseteq V$ and any two of them do not share an edge. In this case, all the transmission links in an IS do not interfere with each other and thus can be carried out successfully at the same time. If adding one more vertex into the IS $I$ results in a non-independent one, the set $I$ is called the maximal independent set (MIS). We can collect all the MISs of the conflict graph in a set ${\cal Q} = \{ {I_1},...,{I_q},...,{I_Q}\}$, where $Q$ represents the total number of MISs, i.e., $Q = |{\cal Q}|$.

\subsubsection{Link Scheduling}
In this paper, we consider different MISs are scheduled with certain time shares (out of unit time) so that the links within each MIS can carry out the transmissions simultaneously. From our previous discussion, we know only one of the MISs can be active at one time instance and we denote the time share allocated to the MIS ${I_q}$ as ${\lambda _q}$. Therefore, we have to satisfy the following constraint
\begin{align}
&\sum\limits_{q = 1}^Q {{\lambda _q}}  \le 1, \label{lnk_sum}\\
&0 \le {\lambda _q} \le 1,\forall q \in \{ 1,...,Q\}. \label{lnk_one}
\end{align}

On the other hand, the link capacity of the link $(i,j)$ can be obtained based on Shannon-Hartley theorem, which is
\begin{equation}\label{mesh_cap}
{c_{i,j}} = {W_m} \cdot {\log _2}(1 + \frac{{{P_t} \cdot {g_{i,j}}}}{{{W_m} \cdot {N_0}}}), \forall i \in \underline {\cal K}, \forall j \in {\cal K}\backslash i,
\end{equation}
where ${W_m}$ is the bandwidth of the harvested band. According to the link scheduling, the actual data rate over the link $(i,j)$ could be 0 if $(i,j)$ is not scheduled at a time instance. In other words, we use ${r_{i,j}}({I_q})$ to represent the achieved data rate over the link $(i,j)$ when ${I_q}$ is scheduled, where ${r_{i,j}}({I_q}) = {c_{i,j}}$ if $(i,j) \in {I_q}$ and 0 otherwise. Considering the link $(i,j)$ could exist in all the MISs, the total achieved data rate over the link $(i,j)$ can be expressed as
\begin{equation}\label{sch_cap}
{R_{i,j}} = \sum\limits_{q = 1}^Q {{\lambda _q}{r_{i,j}}({I_q})}.
\end{equation}

\subsubsection{Flow Routing}
In this paper, we consider the network level end-to-end (from users to the BS) service provisioning. Suppose the end user $u$ initiates a session ${l_u}$, the SSP should determine whether to support it by associating $u$ directly to the BS through one-hop transmission or connecting $u$ to the cognitive mesh network and then arriving at the BS via multi-hop transmissions. Here denote ${f_{i,j}}({l_u})$ as the supported flow rate for the session ${l_u}$ over the link $(i,j)$ at the network level. If node $i$ is the source of session ${l_u}$, which is the end user $u$ (i.e., $s({l_u}) = u$), we have the following constraints
\begin{equation}\label{flow_u_in}
\sum\limits_{j \in \{ j|u \in {{\cal T}_j}\} } {{f_{j,u}}({l_u})} = 0,
\end{equation}
\begin{equation}\label{flow_u_out}
{f_{u,j}}({l_u}) \cdot {x_{u,j}} = r({l_u}).
\end{equation}
The constraint (\ref{flow_u_in}) means that the incoming flow rate of any session at the source node is zero since the end user is the initiator of the session. The constraint (\ref{flow_u_out}) reflects the first hop from the end user $u$ to an infrastructural node $j \in {\cal K}$, where $r({l_u})$ signifies the achievable data rate of user $u$. Clearly, the association variable is coupled with flow rate in (\ref{flow_u_out}), implying that there only exists one wireless link from the end user $u$ to one of the infrastructural nodes to support $u's$ data rate. Besides, the flow rate on a link should be constrained by the link capacity according to (\ref{link_cap}), which is expressed as
\begin{equation}\label{flow_u_bnd}
0 \le {f_{u,j}}({l_u}) \le {c_{u,j}}.
\end{equation}

For any infrastructural node $i \in \underline {\cal K}$, which is the CR router, the flow conservation law (FCL) implies that for any session ${l_u}$, the total flow into $i$ must be equal to the total flow out of $i$. This can be expressed as
\begin{equation}\label{fcl}
\sum\limits_{j \in \{ j|i \in {{\cal T}_j}\} } {{f_{j,i}}({l_u})} + {f_{u,i}}({l_u}) \cdot {x_{u,i}} = \sum\limits_{k \in {{\cal T}_i}} {{f_{i,k}}({l_u})}.
\end{equation}
Clearly, if the CR router $i$ is directly associated with end user $u$, constraint (\ref{fcl}) could be rewritten as $r({l_u}) = \sum\limits_{k \in {{\cal T}_i}} {{f_{i,k}}({l_u})}$ according to (\ref{flow_u_out}); whereas if the CR router $i$ is the intermediate infrastructural node to support ${l_u}$, constraint (\ref{fcl}) would be equivalent to $\sum\limits_{j \in \{ j|i \in {{\cal T}_j}\} } {{f_{j,i}}({l_u})} = \sum\limits_{k \in {{\cal T}_i}} {{f_{i,k}}({l_u})}$.

Moreover, all the flows in the cognitive mesh network are completed at the BS, which means that the BS is the common destination, i.e., $d({l_u}) = b, \forall {l_u} \in {\cal L}$. Thus, we have another constraints for the destination node $b$ described as follows:
\begin{equation}\label{flow_b_out}
\sum\limits_{j \in {{\cal T}_b}} {{f_{b,j}}({l_u})}  = 0,
\end{equation}
\begin{equation}\label{flow_b_in}
{f_{u,b}}({l_u}) \cdot {x_{u,b}} + \sum\limits_{j \in \{ j|b \in {{\cal T}_j}\} } {{f_{j,b}}({l_u})}  = r({l_u}).
\end{equation}
Note that if the end user $u$ is directly connected to the BS (i.e., ${x_{u,b}}=1$), (\ref{flow_b_in}) becomes $\sum\limits_{j \in \{ j|b \in {{\cal T}_j}\} } {{f_{j,b}}({l_u})} = 0$, indicating that session ${l_u}$ is not supported through the CR routers. Instead, if ${x_{u,b}}=0$, meaning user $u$ is associated with the cognitive mesh network, (\ref{flow_b_in}) could be rewritten as $\sum\limits_{j \in \{ j|b \in {{\cal T}_j}\} } {{f_{j,b}}({l_u})}  = r({l_u})$.

Besides, for the link from one CR router $i \in \underline {\cal K}$ to another infrastructural node $j \in {{\cal T}_i}$, the total flow rate on that link (by aggregating all the rates of supported sessions) should not exceed the link capacity. From our previous discussion, we know that the link capacity is dependent on the scheduled time share on that link, which is described as (\ref{sch_cap}). Thus, we have the following constraint
\begin{equation}\label{flow_mesh_bnd}
0 \le \sum\limits_{{l_u} \in {\cal L}} {{f_{i,j}}({l_u})}  \le {R_{i,j}}.
\end{equation}

\subsection{User-centric Network-wide Energy Efficiency Optimization}
Our work aims to maximize the user-centric network-wide (i.e., the weighted sum of end users') energy efficiency by considering user diversity in power capability (e.g., residual energy, maximal allowable power). In light of it, we apply the ratio-based EE model, where the EE is measured in bit/s/Joule and defined as the ratio of data rate to power consumption which in this work are both with respect to the end users. Thus, we coin it as the user-centric EE metric, in contrast to the previous works that applied BS's power consumption in the EE definition \cite{he2014leakage,zarakovitis2016maximizing}. Furthermore, contrary to the conventional EE definition as the ratio of the system sum rate to the sum power consumption \cite{cheung2013achieving,ng2012energy}, our user-centric network-wide EE measures the weighted sum of end users' EE, such that the heterogeneous EE requirements from different users of various power capabilities can be investigated.

Before presenting the optimization problem, we first define the power consumption model of end user $u$ as $\eta {p_u} + {P_c}$, where $\eta$ is the efficiency of the transmit power amplifier when it operates in the linear region, whereas ${P_c}$ is the circuitry power dissipated in all other circuit blocks (e.g., mixer, oscillator, DAC and etc.) which is independent of the transmit power ${p_u}$ and normally a constant value. Then, the EE for the end user $u$ can be obtained as $\frac{{r({l_u})}}{{\eta {p_u} + {P_c}}}$, where $r({l_u})$ denotes its achieved data rate according to (\ref{flow_u_out}). Finally, we introduce a weighting factor ${\omega _u}$ associated with user $u's$ EE, which provides a means for service differentiation as well as fairness. Particularly, the weights could be determined inversely proportional to users' residual energy so that less power capable users are allocated with higher EE priorities.

By considering the user association, power control, channel allocation, routing and link scheduling constraints introduced previously, we can thus formulate the following optimization problem to achieve the maximal \underline{U}ser-centric \underline{N}etwork-wide \underline{EE} (UNEE-Max)
\begin{equation}\label{unee_max}
\begin{aligned}
& {\text{Max}}
& & \sum\limits_{u \in {\cal U}} {{\omega _u}\frac{{r({l_u})}}{{\eta {p_u} + {P_c}}}} \\
& \text{s.t.}
& & (\ref{ass_sum})\sim(\ref{cha_sum}), (\ref{lnk_sum})\sim(\ref{lnk_one}), (\ref{flow_u_in})\sim(\ref{flow_mesh_bnd}); \\
&&& 0 \le {f_{i,j}}({l_u}), \forall {l_u} \in {\cal L}, \forall i \in {\cal K}, \forall j \in {\cal K}\backslash i; \\
&&& 0 \le {p_u} \le {P_{u,\max }}, \forall u \in {\cal U}
\end{aligned}
\end{equation}
where ${x_{u,k}}$, ${N_u}$, ${p_u}$, ${f_{u,j}}({l_u})$, ${f_{i,j}}({l_u})$ and ${\lambda_q}$ are optimization decision variables. Clearly, UNEE-Max is a cross-layer optimization problem involving coupled variables from the physical layer to the network layer. In the next section, we elaborate several difficulties in addressing UNEE-Max problem and introduce techniques to solve it accordingly.

\section{Overview of The UNEE-Max Problem}
\subsection{Complexity of The UNEE-Max}
We first highlight several key difficulties in solving the UNEE-Max problem.
\subsubsection{NP-completeness for searching all MISs}
Under constraint (\ref{lnk_sum}), we need to search all the MISs for link scheduling. However, finding all the MISs in a conflict graph $G = (V,E)$ is NP-complete, which is the common obstacle encountered in multi-hop wireless networks \cite{li2010multi}. Although we can apply brute-force search when the size of $G = (V,E)$ is small, it is highly prohibitive when $G$ becomes large. Therefore, it requires a cost-effective approach to find MISs so as to make the problem tractable.
\subsubsection{Uncertainty of the harvested band}
In CRNs, SUs are allowed to access PUs' spectrum bands only when these bands are not occupied by PUs. SUs must immediately evacuate when PUs reclaim the spectrum. In practice, the availability of these harvested bands is highly unpredictable due to the uncertainty of PUs' activity and SSP's statistical inference model (i.e., false alarm / miss detection probabilities) \cite{yucek2009survey}. Therefore, the bandwidth ${W_m}$ of the harvested band (defined in (\ref{mesh_cap})) is a random variable, whose probability distribution could be derived from statistical characteristics of these PUs' bands from some observations and experiments \cite{mchenry2006chicago}. However, the randomness of ${W_m}$ makes (\ref{flow_mesh_bnd}) a stochastic constraint, which causes the feasible region of UNEE-Max to be both random and nonconvex.
\subsubsection{Combinatorial nature of user association}
The indicator variable ${x_{u,k}}$ in constraints (\ref{ass_sum}) and (\ref{ass_one}) enforces unique association, which makes the problem combinatorial. Although the classical branch-and-bound approach can be applied to solve general integer programming problems, due to the tight coupling between the association and the resource allocation (i.e., power control, link scheduling and routing) in UNEE-Max, it is difficult to solve using traditional approaches.
\subsubsection{Nonconcavity of objective function}
The objective function in the UNEE-Max problem is in the form of weighted sum of linear fractional functions (WSoLFF), which is generally nonconcave \cite{freund2001solving}. An immediate consequence is that the powerful tools from the convex optimization theory do not apply to the UNEE-Max, and the KKT conditions are only necessary conditions for optimality \cite{borwein2010convex}. Therefore, we need to transform the UNEE-Max to a certain form from which approximate solution to the UNEE-Max can be found.

\subsection{The UNEE-Max Relaxation Algorithm}
After outlining the difficulties in solving UNEE-Max problem, we introduce the relaxation or transformation techniques to make the UNEE-Max tractable.
\subsubsection{Critical MIS set}
Although there exists exponentially many MISs in a conflict graph, Li \emph{et. al.} \cite{li2010multi} proved that only a small portion of MISs, termed as \emph{critical MIS set}, can be scheduled in the optimal resource allocation. Instead of searching all MISs, we thus apply the SIO-based approach proposed in \cite{li2010multi,liu2016dafee} to return the critical MIS set ${{\cal Q}^{'}} = \{ {I_1},...,{I_q},...,{I_{{Q^{'}}}}\}$ in polynomial time, where ${{\cal Q}^{'}} \subseteq {\cal Q}$. Therefore, we can replace $Q$ with $Q^{'}$ in constraint (\ref{ass_sum}), and (\ref{flow_mesh_bnd}) to make the UNEE-Max problem more tractable. Note that SIO-based approach may only give a fraction of $Q^{'}$ in a limited searching time leading to a loss in solution optimality. In light of this, we could deliberately allow a longer searching time as the SIO-based approach can be run offline.
\subsubsection{$\Delta$-confidence level}
To address the stochastic constraint (\ref{flow_mesh_bnd}), inspired by the concept of value at risk (VaR) in \cite{holton2003value}, we reformulate it as a chance constraint of $\Delta$-confidence level represented as follows
\begin{equation*}
\Pr \left[ {0 \le \sum\limits_{{l_u} \in {\cal L}} {{f_{i,j}}({l_u})}  \le \sum\limits_{q = 1}^Q {{\lambda _q}{W_m}{{\log }_2}(1 + \frac{{{P_t} \cdot {g_{i,j}}}}{{{W_m}{N_0}}})} } \right] \ge \Delta,
\end{equation*}
where $\Delta  \in [0,1]$ indicates the confidence level for stochastic constraint (\ref{flow_mesh_bnd}) to be satisfied and $(i,j) \in {I_q}$. Suppose ${F_{{W_m}}}( \cdot )$ represent the cumulative distribution function (CDF) of random variable (r.v.) ${W_m}$. We could then obtain the ${F_{{c_{i,j}}}}( \cdot )$ as the CDF for the r.v. ${c_{i,j}} = {W_m}{\log _2}(1 + \frac{{{P_t} \cdot {g_{i,j}}}}{{{W_m}{N_0}}})$, which is the link capacity of $(i,j)$. Thus, by integrating the critical MISs, the above inequality could be reformulated as
\begin{equation}\label{flow_mesh_bnd_new}
0 \le \sum\limits_{{l_u} \in {\cal L}} {{f_{i,j}}({l_u})}  \le \sum\limits_{q = 1}^{{Q^{'}}} {{\lambda _q}F_{{c_{i,j}}}^{ - 1}(1 - \Delta )}.
\end{equation}

By replacing (\ref{flow_mesh_bnd}) with (\ref{flow_mesh_bnd_new}), the original stochastic constraint is converted to a linear inequality constraint in ${f_{i,j}}({l_u})$ and ${\lambda _q}$.
\subsubsection{Integer relaxation and rounding}
In the first phase, we assume that end users can be associated with the BS and CR routers at the same time. In other words, we relax the integer association variable ${x_{u,k}}$ to the continuous domain of $[0,1]$. Under this assumption, we also introduce the sub-channel auxiliary variable ${N_{u,k}} = {N_u} \cdot {x_{u,k}}$, where ${N_{u,k}} \in \mathbb{R}^{+}$ and $\mathbb{R}^{+}$ represents set of all nonnegative numbers, the power auxiliary variable ${p_{u,k}} = {p_u} \cdot {x_{u,k}}$ and the flow auxiliary variable ${\widetilde f_{u,k}}({l_u}) = {f_{u,k}}({l_u}) \cdot {x_{u,j}}$, so that $\sum\limits_{k \in {\cal K}} {{N_{u,k}}}  = {N_u}$, $\sum\limits_{k \in {\cal K}} {{p_{u,k}}}  = {p_u}$, $\sum\limits_{k \in {\cal K}} {{{\widetilde f}_{u,j}}({l_u})}  = r({l_u})$. Therefore, we can eliminate association constraint (\ref{ass_sum}) and (\ref{ass_one}) and rewrite the UNEE-Max problem as a Relaxed-UNEE-Max problem which is described as follows
\begin{equation}\label{relaxed-unee-max}
\begin{aligned}
& {\text{Max}}
& & \sum\limits_{u \in {\cal U}} {{\omega _u}\frac{{\sum\limits_{k \in {\cal K}} {{{\widetilde f}_{u,k}}({l_u})} }}{{\eta \sum\limits_{k \in {\cal K}} {{p_{u,k}}}  + {P_c}}}} \\
& \text{s.t.}
& & 0 \le \sum\limits_{u \in {\cal U}} {\sum\limits_{k \in {\cal K}} {{N_{u,k}}} }  \le {N_{tot}}, {N_{u,k}} \in \mathbb{R}^{+}; \\
&&& 0 \le \sum\limits_{k \in {\cal K}} {{p_{u,k}}}  \le {P_{u,\max }}, \forall u; {p_{u,k}} \in \mathbb{R}^{+}; \\
&&& 0 \le {\widetilde f_{u,k}}({l_u}) \le {N_{u,k}}W{\log _2}(1 + \frac{{{p_{u,k}}|{h_{u,k}}{|^2}}}{{{N_{u,k}}W{N_0}}}), \forall u, \forall k; \\
&&& \sum\limits_{j \in \{ j|i \in {{\cal T}_j}\} } {{f_{j,i}}({l_u})}  + {\widetilde f_{u,i}}({l_u}) = \sum\limits_{k \in {{\cal T}_i}} {{f_{i,k}}({l_u})}, \forall u; \\
&&& \sum\limits_{j \in \{ j|b \in {{\cal T}_j}\} } {{f_{j,b}}({l_u})}  = \sum\limits_{k \in \underline {\cal K} } {{{\widetilde f}_{u,k}}({l_u})}, \forall u; \\
&&& (\ref{lnk_sum})\sim(\ref{lnk_one}), (\ref{flow_u_in}), (\ref{flow_b_out}), (\ref{flow_mesh_bnd_new}); \\
&&& 0 \le {f_{i,j}}({l_u}), \forall {l_u} \in {\cal L}, \forall i \in {\cal K}, \forall j \in {\cal K}\backslash i.
\end{aligned}
\end{equation}

In the second phase, we develop a rounding scheme, as what will be discussed in Section VII, to convert the output of the Relaxed-UNEE-Max problem into a feasible value that satisfies the constraints of original problem UNEE-Max in (\ref{unee_max}).

\subsubsection{Parametric subtractive transformation}
It is clear that the constraints in the Relaxed-UNEE-Max problem (\ref{relaxed-unee-max}) form a convex feasible set $\bm{{\cal X}}$ w.r.t. the optimization variable set $(\bm{p}, \bm{N}, \bm{\widetilde f}, \bm{f}, \bm{\lambda}) \in \bm{{\cal X}}$. \footnote{For the sake of brevity, we define the vectors of optimization variables as $\bm{p} = \{ {p_{u,k}}\}$, $\bm{N} = \{ {N_{u,k}}\}$, $\bm{\widetilde f} = \{ {\widetilde f_{u,k}}({l_u})\}$, $\bm{f} = \{ {f_{i,j}}({l_u})\}$ and $\bm{\lambda}  = \{ {\lambda _q}\}$.} However, it is still challenging due to the sum-of-ratio form in the objective \cite{schaible2003fractional}. To overcome this difficulty, we firstly transform the objective function in (\ref{relaxed-unee-max}) into an intermediate form by introducing an auxiliary variable $\bm{\alpha} = \{ \alpha ,...,{\alpha _U}\}$ and reformulate the Relaxed-UNEE-Max problem as
\begin{equation}\label{int-relaxed-unee-max}
\begin{aligned}
& {\text{Max}}
& & \sum\limits_{u \in {\cal U}} {{\omega_u}{\alpha_u}} \\
& \text{s.t.}
& & \frac{{\sum\limits_{k \in {\cal K}} {{{\widetilde f}_{u,k}}({l_u})} }}{{\eta \sum\limits_{k \in {\cal K}} {{p_{u,k}}}  + {P_c}}} \ge {\alpha _u}, \forall u; \\
&&& (\bm{p}, \bm{N}, \bm{\widetilde f}, \bm{f}, \bm{\lambda}) \in \bm{{\cal X}}.
\end{aligned}
\end{equation}
Although the objective is an affine function w.r.t. $\bm{\alpha}$, problem (\ref{int-relaxed-unee-max}) is not a convex optimization yet due to the fractional constraint. Thus, we further convert (\ref{int-relaxed-unee-max}) into a parametric subtractive form and show in the following theorem its equivalence to the weighted sum maximization problem (\ref{int-relaxed-unee-max}) with fractional constraint.

\begin{theorem}\label{theory}
Suppose $(\bm{p^{*}}, \bm{N^{*}}, \bm{\widetilde f^{*}}, \bm{f^{*}}, \bm{\lambda^{*}}, \bm{\alpha^{*}})$ is the solution to problem (\ref{int-relaxed-unee-max}), there exist $\bm{\beta^{*}}$ such that for the parametric variables $\bm{\alpha}=\bm{\alpha^{*}}$ and $\bm{\beta}=\bm{\beta^{*}}$, $(\bm{p^{*}}, \bm{N^{*}}, \bm{\widetilde f^{*}}, \bm{f^{*}}, \bm{\lambda^{*}})$ satisfies the KKT conditions of the following problem
\begin{equation}\label{p-relaxed-unee-max}
\begin{aligned}
& {\text{Max}}
& & \sum\limits_{u \in {\cal U}} {{\beta _u}\left[ {\sum\limits_{k \in {\cal K}} {{{\widetilde f}_{u,k}}({l_u})}  - {\alpha _u}(\eta \sum\limits_{k \in {\cal K}} {{p_{u,k}}}  + {P_c})} \right]} \\
& \text{s.t.}
& & (\bm{p}, \bm{N}, \bm{\widetilde f}, \bm{f}, \bm{\lambda}) \in \bm{{\cal X}}.
\end{aligned}
\end{equation}
Also, the following system equations hold for the parametric variables ($\bm{\alpha^{*}}, \bm{\beta^{*}}$) and the tuple $(\bm{p^{*}}, \bm{N^{*}}, \bm{\widetilde f^{*}}, \bm{f^{*}}, \bm{\lambda^{*}})$:
\begin{equation}\label{parameter}
\left\{
\begin{aligned}
&{\alpha _u} = \frac{{\sum\limits_{k \in {\cal K}} {{{\widetilde f}_{u,k}}({l_u})} }}{{\eta \sum\limits_{k \in {\cal K}} {{p_{u,k}}} + {P_c}}} \\
&{\beta _u} = \frac{{{\omega _u}}}{{\eta \sum\limits_{k \in {\cal K}} {{p_{u,k}}}  + {P_c}}}.
\end{aligned}
\right.
\end{equation}

On the contrary, if $(\bm{p^{*}}, \bm{N^{*}}, \bm{\widetilde f^{*}}, \bm{f^{*}}, \bm{\lambda^{*}})$ is the solution to problem (\ref{p-relaxed-unee-max}) while (\ref{parameter}) system equations are met for $\bm{\alpha}=\bm{\alpha^{*}}$ and $\bm{\beta}=\bm{\beta^{*}}$, then $(\bm{p^{*}}, \bm{N^{*}}, \bm{\widetilde f^{*}}, \bm{f^{*}}, \bm{\lambda^{*}}, \bm{\alpha^{*}})$ satisfies KKT conditions for problem (\ref{int-relaxed-unee-max}), where $\bm{\beta}=\bm{\beta^{*}}$ is the Lagrange multiplier for fractional constraint in (\ref{int-relaxed-unee-max}).
\end{theorem}
\begin{proof}
See Appendix A.
\end{proof}

Based on Theorem \ref{theory}, we can address the problem (\ref{int-relaxed-unee-max}) by solving (\ref{p-relaxed-unee-max}) while guaranteeing (\ref{parameter}), such that the solution of the Relaxed-UNEE-Max could be obtained. Furthermore, it is worth noting that if the solution is unique, it is also the global solution \cite{freund2001solving}. Toward solving (\ref{p-relaxed-unee-max}), we apply a dual-based approach, which has been widely adopted in various network settings for its simplicity of implementation, and augment it with the parametric programming to form inner loop and outer loop iterative update processes. The detailed steps are described in the following section.

\section{Algorithm for The Relaxed-UNEE-Max}
Based on our prior discussion, the Relaxed-UNEE-Max problem is equivalent to problem (\ref{int-relaxed-unee-max}), whose solution is identical to (\ref{p-relaxed-unee-max}) when satisfying (\ref{parameter}). Hence, we focus on solving problem (\ref{p-relaxed-unee-max}), and for the presentation clarity, we first outline the general idea of the solution algorithm.

The whole algorithm is split into an inner loop and an outer loop optimization problem. The algorithm starts with initializing the parametric variables $\bm{\alpha}$ and $\bm{\beta}$. For the given $\bm{\alpha}$ and $\bm{\beta}$, (\ref{p-relaxed-unee-max}) becomes a convex optimization problem with an affine objective and a convex feasible set. The inner loop applies dual decomposition approach to solve this convex optimization problem, and each iteration of the dual-based method is termed as the inner loop iterations. Multiple inner loop iterations are performed till the optimal dual and primal solutions are reached. The output of inner loop, which is $(\bm{p^{*}}, \bm{N^{*}}, \bm{\widetilde f^{*}}, \bm{f^{*}}, \bm{\lambda^{*}})$, are then fed back to the outer loop to update the parametric variables $\bm{\alpha}$ and $\bm{\beta}$. The overall algorithm terminates if the convergence condition for $\bm{\alpha}$ and $\bm{\beta}$ (we will elaborate it later) are met. Otherwise, the algorithm continues by solving the inner loop optimization problem again using the updated $\bm{\alpha}$ and $\bm{\beta}$.
\subsection{Algorithm for The Inner Loop Optimization Problem}
Suppose the parametric variables are $\bm{\alpha^{t}}$ and $\bm{\beta^{t}}$ at the ${t^{th}}$ outer loop iteration, the inner loop procedure starts with introducing a partial Lagrange multiplier $\bm{v} = \{ {v_{1,1}},...,{v_{1,|{\cal K}|}},...,{v_{U,|{\cal K}|}}\}$ w.r.t. the third constraint (nonlinear capacity constraint) in problem (\ref{relaxed-unee-max}) attempting to decouple the decision variables. We denote the partial Lagrangian by $\bm{{\cal L}}(\bm{v}; \bm{p}, \bm{N}, \bm{\widetilde f}, \bm{f}, \bm{\lambda})$ and express it as
\begin{equation}
\begin{aligned}
\bm{{\cal L}}(\bm{v}; &\bm{p}, \bm{N}, \bm{\widetilde f}, \bm{f}, \bm{\lambda}) =  \\
&\sum\limits_{u \in {\cal U}} {{\beta_u^t}\left[ {\sum\limits_{k \in {\cal K}} {{{\widetilde f}_{u,k}}({l_u})}  - {\alpha _u}(\eta \sum\limits_{k \in {\cal K}} {{p_{u,k}}}  + {P_c})} \right]}  - \\
&\sum\limits_{u \in {\cal U}} {\sum\limits_{k \in {\cal K}} {{v_{u,k}}\left[ {{{\widetilde f}_{u,k}}({l_u}) - {N_{u,k}}W{{\log }_2}(1 + \frac{{{p_{u,k}}|{h_{u,k}}{|^2}}}{{{N_{u,k}}W \cdot {N_0}}})} \right]} }.
\end{aligned}
\end{equation}

The dual function can be then obtained as
\begin{equation*}
{\cal D}(\bm{v}) = \mathop {\max }\limits_{\bm{p}, \bm{N}, \bm{\widetilde f}, \bm{f}, \bm{\lambda}} \bm{{\cal L}}(\bm{v}; \bm{p}, \bm{N}, \bm{\widetilde f}, \bm{f}, \bm{\lambda}).
\end{equation*}
Since problem (\ref{p-relaxed-unee-max}) is convex and Slater’s condition for constraint qualification is assumed to hold, it follows that there is no duality gap and thus the primal problem can be solved via its dual
\begin{equation*}
\text{Relaxed-UNEE-Max Optimal} = \mathop {\min }\limits_{\bm{v} \succeq  0} {\cal D}(\bm{v}).
\end{equation*}

\subsubsection{Dual problem}
We solve the dual variables via the projected subgradient method. First, let us denote the primal variables obtained at ${s^{th}}$ inner loop iteration as $(\bm{p^{s}}, \bm{N^{s}}, \bm{\widetilde f^{s}}, \bm{f^{s}}, \bm{\lambda^{s}})$. Then, the dual variables at ${s^{th}}$ inner loop iteration are updated as follows
\begin{equation}\label{dual_update}
v_{u,k}^{s + 1} = {\left[ {v_{u,k}^s + \delta (\widetilde {f_{u,k}^s}({l_u}) - N_{u,k}^sW{{\log }_2}(1 + \frac{{p_{u,k}^s|{h_{u,k}}{|^2}}}{{N_{u,k}^sW \cdot {N_0}}}))} \right]^ + }
\end{equation}
where $\delta$ is the step size and $[\cdot]^{+}$ denotes the projection into the set of non-negative real numbers.

In what follows, we focus on solving the primal variables given the dual variables at each inner loop iteration.
\subsubsection{Primal problem}
We now argue that the primal problem
\begin{equation*}
\arg \mathop {\max }\limits_{\bm{p}, \bm{N}, \bm{\widetilde f}, \bm{f}, \bm{\lambda}} \bm{{\cal L}}(\bm{v^{s}}; \bm{p}, \bm{N}, \bm{\widetilde f}, \bm{f}, \bm{\lambda})
\end{equation*}
can be reorganized into a routing subproblem and a physical layer resource allocation subproblem. Thus, solving the primal problem is equivalent to solving two independent subproblems, each of which is fairly straightforward. Toward this end, we rewrite original partial Lagrangian as follows
\begin{equation*}
\begin{aligned}
&\bm{{\cal L}}(\bm{v^{s}}; \cdot) = \underbrace {\sum\limits_{u \in {\cal U}} {\sum\limits_{k \in {\cal K}} {\left[ \beta _u^t {{{\widetilde f}_{u,k}}({l_u}) - {v_{u,k}^s}{{\widetilde f}_{u,k}}({l_u})} \right]} } }_\text{routing subproblem} + \sum\limits_{u \in {\cal U}} \Big[ \\
&\underbrace {\sum\limits_{k \in {\cal K}} {{v_{u,k}^s}{N_{u,k}}W{{\log}_2}(1+\frac{{{p_{u,k}}|{h_{u,k}}{|^2}}}{{{N_{u,k}}W{N_0}}})}-\beta _u^t\alpha_u^t(\eta \sum\limits_{k \in {\cal K}} {{p_{u,k}}}+{P_c}) }_\text{resource allocation subproblem} \Big]
\end{aligned}
\end{equation*}
and we can represent it as $\bm{{\cal L}}(\bm{v^{s}}; \cdot) = \bm{{\cal L}}(\bm{v^{s}}; \bm{\widetilde f}, \bm{f}, \bm{\lambda})_{rout} + \bm{{\cal L}}(\bm{v^{s}}; \bm{p}, \bm{N})_{res}$. Thus, we can separate the primal optimization problem into the following subproblems
\begin{equation}\label{routing-lag}
\begin{aligned}
& {\text{Max}}
& & \bm{{\cal L}}(\bm{v^{s}}; \bm{\widetilde f}, \bm{f}, \bm{\lambda})_{rout} \\
& \text{s.t.}
& & \sum\limits_{j \in \{ j|i \in {{\cal T}_j}\} } {{f_{j,i}}({l_u})}  + {\widetilde f_{u,i}}({l_u}) = \sum\limits_{k \in {{\cal T}_i}} {{f_{i,k}}({l_u})}, \forall u; \\
&&& \sum\limits_{j \in \{ j|b \in {{\cal T}_j}\} } {{f_{j,b}}({l_u})}  = \sum\limits_{k \in \underline {\cal K} } {{{\widetilde f}_{u,k}}({l_u})}, \forall u; \\
&&& (\ref{lnk_sum})\sim(\ref{lnk_one}), (\ref{flow_u_in}), (\ref{flow_b_out}), (\ref{flow_mesh_bnd_new}); \\
&&& 0 \le {f_{i,j}}({l_u}), \forall {l_u} \in {\cal L}, \forall i \in {\cal K}, \forall j \in {\cal K}\backslash i.
\end{aligned}
\end{equation}
which is the routing subproblem, while
\begin{equation}\label{res-lag}
\begin{aligned}
& {\text{Max}}
& & \bm{{\cal L}}(\bm{v^{s}}; \bm{p}, \bm{N})_{res} \\
& \text{s.t.}
& & 0 \le \sum\limits_{u \in {\cal U}} {\sum\limits_{k \in {\cal K}} {{N_{u,k}}} }  \le {N_{tot}}, {N_{u,k}} \in \mathbb{R}^{+}; \\
&&& 0 \le \sum\limits_{k \in {\cal K}} {{p_{u,k}}}  \le {P_{u,\max }}, \forall u; {p_{u,k}} \in \mathbb{R}^{+},
\end{aligned}
\end{equation}
which is the physical layer resource allocation (i.e., power control and channel allocation) subproblem.

It is clear that with the given dual variables $\bm{v^{s}}$ and the parametric variables $\bm{\alpha^{t}}$ and $\bm{\beta^{t}}$, the routing subproblem (\ref{routing-lag}) is a linear optimization problem w.r.t. the decision variables, which can be easily solved by many softwares, such as CPLEX. On the other hand, the resource allocation subproblem belongs to the general convex optimization problem with a concave objective and a convex feasible region. Thus, it also can be easily solved via the interior point method, for instance.

With the primal variables obtained at each iteration, they are fed back to the dual variable update process according to (\ref{dual_update}), and we keep iterating the inner loop iterations till a predefined stopping criterion is met.
\subsubsection{Stopping criterion and step size}
First, we define the stopping criterion for the inner loop algorithm as $|\bm{v^{s+1}} - \bm{v^{s}}| \le \varepsilon$, where $\varepsilon$ denotes a predefined threshold. On the other hand, the choice of step size $\delta$ affects the convergence rate of the solution. Normally, we could apply diminishing step size or constant but sufficiently small step size \cite{bertsekas1999nonlinear}, which are both guaranteed to converge to the optimal solutions. We will examine the impact of step size on the convergence rate in the performance evaluation section.

\subsection{Algorithm for The Outer Loop Optimization Problem}
The outer loop optimization problem is in a parametric subtractive form as the objective in problem (\ref{p-relaxed-unee-max}). The goal is to iteratively obtain the parametric variables $\bm{\alpha}$ and $\bm{\beta}$, where the iteration here is termed as the outer loop iteration. Parameter $\bm{\alpha}$ may be intuitively viewed as the ``price'' of power consumption while parameter $\bm{\beta}$ is introduced as the Lagrange multiplier for the fractional constraint in (\ref{int-relaxed-unee-max}). Here, we apply the gradient method \cite{he2014leakage} to update the parametric variables in a following way:
\begin{equation}\label{alpha_update}
\alpha _u^{t + 1} = \alpha _u^t - \xi (\alpha _u^t - \frac{{\sum\limits_{k \in {\cal K}} {\widetilde {f_{u,k}^{t,{s^*}}}({l_u})} }}{{\eta \sum\limits_{k \in {\cal K}} {p_{u,k}^{t,{s^*}}}  + {P_c}}}), \forall u,
\end{equation}
\begin{equation}\label{beta_update}
\beta _u^{t + 1} = \beta _u^t - \xi (\beta _u^t - \frac{{{\omega _u}}}{{\eta \sum\limits_{k \in {\cal K}} {p_{u,k}^{t,{s^*}}}  + {P_c}}}), \forall u,
\end{equation}
where $\widetilde {f_{u,k}^{t,{s^*}}}({l_u})$ and $p_{u,k}^{t,{s^*}}$ are the converged values of decision variables after ${s^{*}}$ inner loop iterations. Similar to the inner loop optimization, another small threshold value $\sigma$ is selected and the stopping criterion is set to $|\bm{\alpha^{t+1}} - \bm{\alpha^{t}}| \le \sigma$ and $|\bm{\beta^{t+1}} - \bm{\beta^{t}}| \le \sigma$. The convergence of the outer loop optimization can be guaranteed by the gradient method and the step size $\xi$ should be selected to be sufficiently small. Later, we will give the convergence analysis in the performance evaluation section.
\begin{figure}[!htb]
  \begin{center}
  \includegraphics[width=3.2in]{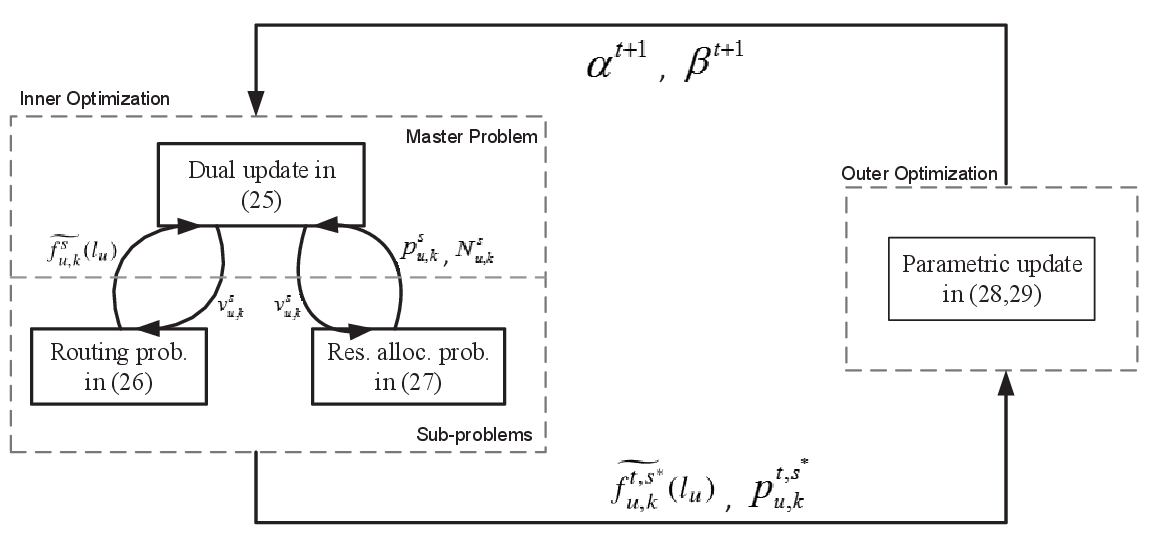}
  \end{center}
  \begin{center}
   \parbox{8cm}{\caption{Summary diagram for the solution algorithm of problem (\ref{p-relaxed-unee-max}).}\label{opt_flow}}
  \end{center}
\end{figure}

For the presentation clarity, we give a high level overview of the solution algorithm for optimization problem (\ref{p-relaxed-unee-max}) as shown in Fig.\ref{opt_flow}, which shows the necessary information exchange between solution processes. Besides, Algorithm \ref{opt_algo} formally describes the solution algorithm for the Relaxed-UNEE-Max.
\begin{algorithm}
\caption{Algorithm for Solving Relaxed-UNEE-Max}\label{opt_algo}
\begin{algorithmic}[1]
    \REQUIRE Given network settings; Initialize all the variables $\bm{p^{0}}, \bm{N^{0}}, \bm{\widetilde f^{0}}$ to any feasible value; Let $\bm{v^{0}}=\bm{\alpha^{0}}=\bm{\alpha^{0}}=0$; Set $t=s=0$; Initialize thresholds $\sigma, \varepsilon$ and step size $\delta, \xi$.
    \ENSURE $\bm{p^{*}}, \bm{N^{*}}, \bm{\widetilde f^{*}}, \bm{f^{*}}, \bm{\lambda^{*}}$
\STATE Calculate $\bm{\alpha^{1}}$, $\bm{\beta^{1}}$ and $\bm{v^{1}}$ according to Eq.(\ref{parameter}) and Eq.(\ref{dual_update}), respectively.
\WHILE{$|\bm{\alpha^{t+1}} - \bm{\alpha^{t}}| \geq \sigma$ or $|\bm{\beta^{t+1}} - \bm{\beta^{t}}| \geq \sigma$}
\STATE $t \leftarrow t+1$;
\WHILE{$|\bm{v^{s+1}} - \bm{v^{s}}| \geq \varepsilon$}
\STATE $s \leftarrow s+1$;
\STATE Solve resource allocation sub-problem (\ref{res-lag}) and obtain $\bm{p^{s}}, \bm{N^{s}}$;
\STATE Solve routing sub-problem (\ref{routing-lag}) and obtain $\bm{\widetilde f^{s}}, \bm{f^{s}}, \bm{\lambda^{s}}$;
\STATE Update dual variable $\bm{v^{s+1}}$ according to Eq.(\ref{dual_update});
\ENDWHILE
\STATE Update parametric variables $\bm{\alpha^{t+1}}$ and $\bm{\beta^{t+1}}$ according to Eq.(\ref{alpha_update}) and Eq.(\ref{beta_update}), respectively;
\ENDWHILE
\end{algorithmic}
\end{algorithm}
\section{User Association and Integer Rounding}
To this end, the problem (\ref{p-relaxed-unee-max}) is solved via the prior algorithm whose solution is identical to the one in the Relaxed-UNEE-Max problem (\ref{relaxed-unee-max}). However, due to the physical constraint that every user can only be associated with one infrastructural node, the previously obtained solution should be converted to a feasible one for the original problem. Besides, the integer property of the number of allocated OFDM sub-channels also requires a further rounding procedure to the obtained solution. Inevitably, this step could introduce performance degradation, but in the performance evaluation section, we will show that its impact on the performance is quite limited.

First, we present the association rule as
\begin{equation*}
k = \arg \, \mathop {\max }\limits_{i \in {\cal K}} \frac{{\widetilde {f_{u,i}^*}({l_u})}}{{\eta p_{u,i}^* + {P_c}}} \,, \forall u.
\end{equation*}
The above operation indicates that we associate the user with the infrastructural node which provides the largest value of EE. In other words, if end user $u$ obtains the highest EE from node $k$, we set the association variable ${x_{u,k}}=1$ while ${x_{u,i}}=0$ for $i \ne k$. In so doing, we can fix the association variables and the original problem UNEE-Max in (\ref{unee_max}) could be simplified significantly. Here, we coin this simplified problem by fixing the association variables as \emph{Asso-UNEE-Max} and it can be similarly addressed by the prior algorithm in Fig.\ref{opt_flow}. In later section, the comparison between the network performance of Asso-UNEE-Max and the one obtained by solving Relaxed-UNEE-Max will be demonstrated.

Next, we introduce the integer rounding function as
\begin{equation*}
Rnd({N_u}) = \max \{ \left\lfloor {{N_u}} \right\rfloor ,1\} \,, \forall u,
\end{equation*}
where the operator $\left\lfloor  \cdot  \right\rfloor$ rounds the input to the greatest integer that is less than or equal to the input. Besides, the reason we apply $max$ function is to guarantee that every end user can at least be assigned with one sub-channel for fairness. The rounding operation is applied to the solution obtained from solving the Asso-UNEE-Max problem, so that the OFDM channel allocation can be determined accordingly. However, the flow variables obtained from Asso-UNEE-Max may not be feasible anymore when doing integer rounding. Therefore, we need to re-solve the UNEE-Max problem (\ref{unee_max}) and get the calibrated flow values which are the feasible ones. Noticing that for the fixed channel allocation, sub-problem (\ref{res-lag}) can be easily addressed by classical iterative water-filling algorithm \cite{he2013water}, which is just a one-dimensional (i.e., power) optimization problem. Here, we denote this solution as the one from a so-called \emph{Rnd-UNEE-Max} problem. Its performance will be compared with the ones obtained from Asso-UNEE-Max and Relaxed-UNEE-Max, respectively, in the evaluation section.
\begin{algorithm}
\caption{Algorithm for User Association and Integer Rounding of Outputs of Algorithm \ref{opt_algo}}\label{rnd_algo}
\begin{algorithmic}[1]
    \REQUIRE Given the output of Algorithm \ref{opt_algo}.
    \ENSURE The calibrated variables $\bm{p^{*}}, \bm{N^{*}}, \bm{\widetilde f^{*}}, \bm{f^{*}}, \bm{\lambda^{*}}$
\FOR {u=1:U}
\STATE Find $k$ such that $k = \arg \, \mathop {\max }\limits_{i \in {\cal K}} \frac{{\widetilde {f_{u,i}^*}({l_u})}}{{\eta p_{u,i}^* + {P_c}}}$;
\STATE Set ${x_{u,k}}=1$;
\ENDFOR
\STATE Update the problem (\ref{unee_max}) and solve the Relaxed-UNEE-Max according to Alg.\ref{opt_algo} to obtain $\bm{p^{'}}, \bm{N^{'}}, \bm{\widetilde f^{'}}, \bm{f^{'}}, \bm{\lambda^{'}}$;
\FOR {u=1:U}
\STATE Let ${N_u}^{*} = \max \{ \left\lfloor {{N_u}^{'}} \right\rfloor ,1\}$;
\ENDFOR
\STATE Update the problem (\ref{unee_max}) and solve the Relaxed-UNEE-Max according to Alg.\ref{opt_algo} to obtain $\bm{p^{*}}, \bm{\widetilde f^{*}}, \bm{f^{*}}, \bm{\lambda^{*}}$.
\end{algorithmic}
\end{algorithm}

In Algorithm \ref{rnd_algo}, we formally give the detailed steps to describe the algorithm for user association and integer rounding for the outputs of Algorithm \ref{opt_algo}.

\section{Performance Evaluation}
\subsection{Simulation Setup}
We consider a $500\times500 {m^2}$ area served by one BS and 12 CR routers, where the BS is put in the center while CR routers represented in squares are placed around it, as shown in Fig.\ref{topology}. We also randomly scatter 35 end users in this area whose locations are shown by dots. The end users' devices are assumed to have an identical circuitry power consumption ${P_c}=50 mW$ and power amplifier efficiency $\eta=5.78$. We assume the users' allowable transmit power ${P_{u, \max}}$ may vary and its impact on system performance will be examined later. To provide fairness for end users, all the weighting factors $\omega$ are set to 1. On the infrastructure side, the CR routers are assumed to employ fixed power ${P_t}$ for transmission and their antenna gain is set to $\zeta=4.63$. The power interference threshold $P_I^{th}$ is set to $3.59\times10^{-7} W$ while the receiving power threshold $P_r^{th}$ is set to $1.0\times10^{-6} W$. The transmission environment between infrastructural nodes are assumed to have path loss exponent $n=3$. Given these system parameters, the interference/communication range can be calculated numerically and we could obtain the conflict graph in this regard. We utilize the OFDM channel model for wireless link between end user $u$ and infrastructural node $k$ as $128.1 + 37.6{\log _{10}}({r_{u,k}})$ dBm where ${r_{u,k}}$ is in kilometers \cite{3gpp2010further}. Following the standard, we set the bandwidth of each sub-channel as $W=180 KHz$. The noise power spectral density is set to ${N_0} = 1 \times {10^{ - 12}} W/Hz$. In addition, for the harvested band, we consider that the availability of it follows a uniform distribution.
\begin{figure}[!htb]
  \begin{center}
  \includegraphics[width=2.8in]{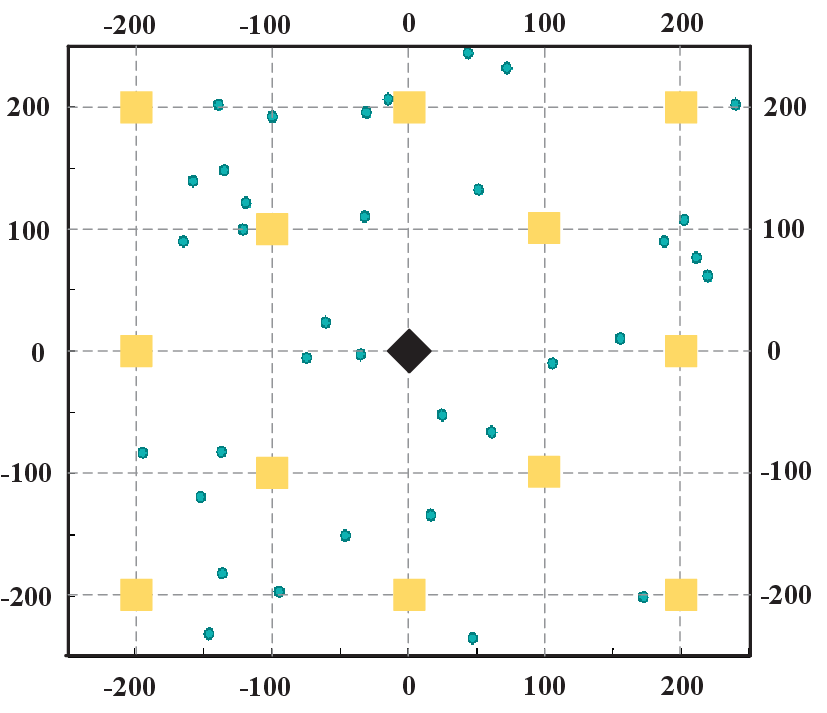}
  \end{center}
  \begin{center}
   \parbox{8cm}{\caption{Evaluated network topology in an $500\times500 {m^2}$ area: 35 end users in blue dots, 12 CR routers in yellow squares and 1 BS in black diamond.}\label{topology}}
  \end{center}
\end{figure}

As for the algorithmic parameter settings, we set the stopping threshold $\varepsilon$ and $\sigma$ as 0.01 and 0.8, respectively; while the step size $\delta=\xi=1 \times {10^{ - 5}}$.
\subsection{Benchmark Setting}
To demonstrate the advantage of our proposed ideology in improving end users' energy efficiency, we leverage the basic cellular network (i.e., 4G/LTE) as the benchmark to compare with. In other words, we consider the same end users within this geographical area served by the small cell BS as shown in Fig.\ref{topology} (excluding CR routers). Similarly, the benchmark UNEE maximization problem, coined as Ben-UNEE-Max, can be proposed as follows:
\begin{equation}\label{ben_unee_max}
\begin{aligned}
& {\text{Max}}
& & \sum\limits_{u \in {\cal U}} {{\omega _u}\frac{{r({l_u})}}{{\eta {p_u} + {P_c}}}} \\
& \text{s.t.}
& & \sum\limits_{u \in {\cal U}} {{N_u}}  \le {N_{tot}}, {N_u} \in {\mathbb{Z}^ + }; \\
&&& 0 \le {r({l_u})} \le {N_{u}}W{\log _2}(1 + \frac{{{p_{u}}|{h_{u,b}}{|^2}}}{{{N_{u}}W \cdot {N_0}}}), \forall u \in {\cal U}; \\
&&& 0 \le {p_u} \le {P_{u,\max }}, \forall u \in {\cal U}.
\end{aligned}
\end{equation}

In this benchmark setting, end users have to be served by the BS so the design of user association, link scheduling and flow routing is eliminated. Rather, we only consider the power control and channel allocation in this one-hop transmission scenario. To effectively solve (\ref{ben_unee_max}), the same transformation approach can be applied to firstly convert it into a tractable one, which is then solved via the water-filling algorithm \cite{he2013water}.

\subsection{Results and Analysis}
First, we examine the convergence behaviors for both inner loop and outer loop optimizations. Since the inner loop is a dual-based (i.e., Lagrangian) algorithm, we also compare its convergence rate under different selection of step sizes. The results are shown in Fig.\ref{cov_analysis}. For demonstrative purposes, we only randomly select 10 users for this simulation and use their average EE as the metric to show the convergence performance. Here, the user's maximum allowable power ${P_{u, \max}}$ is set as 1.5W while CR routers' transmit power ${P_t}$ is set as 1W. The bandwidth of harvested band ${W_m}$ is 100KHz and we set the confidence level $\Delta=0.7$, while the number of OFDM sub-channels is selected to ${N_{tot}}=100$. Moreover, the data in Fig.\ref{inner_iter} is collected at the last iteration of the outer loop optimization.

As we can see from Fig.\ref{inner_iter}, the average EE monotonically increases till the algorithm converges and the EE remains relatively constant (i.e., the difference not exceeding the threshold) afterwards. It can be observed that the algorithm can be guaranteed to converge to the same value under three different step-size settings, but $\delta=1 \times {10^{ - 5}}$ gives the fastest convergence rate (around 48 iterations). This is because as long as the step size is sufficiently small to guarantee convergence, an even smaller step size is not necessary as it slows down the rate to the optimal value. On the other hand, Fig.\ref{outer_iter} illustrates the convergence performance for the outer loop algorithm. For the notational convenience, we take the reciprocal of $\beta$ so that its unit now becomes $W$, while the unit of $\alpha$ is naturally being $Kbits/J$ according to (\ref{parameter}). It can be seen that the converged optimal value of $\alpha$ is exactly the same as the one in Fig.\ref{inner_iter}, which proves the overall convergence of Algorithm \ref{opt_algo}. On the other hand, we observe that the average transmit power for end users is around 0.28W at convergence, which is a small value compared to ${P_{u, \max}}$.
\begin{figure}[!t]
\begin{subfigure}[t]{0.225\textwidth}
  \includegraphics[width=\linewidth]{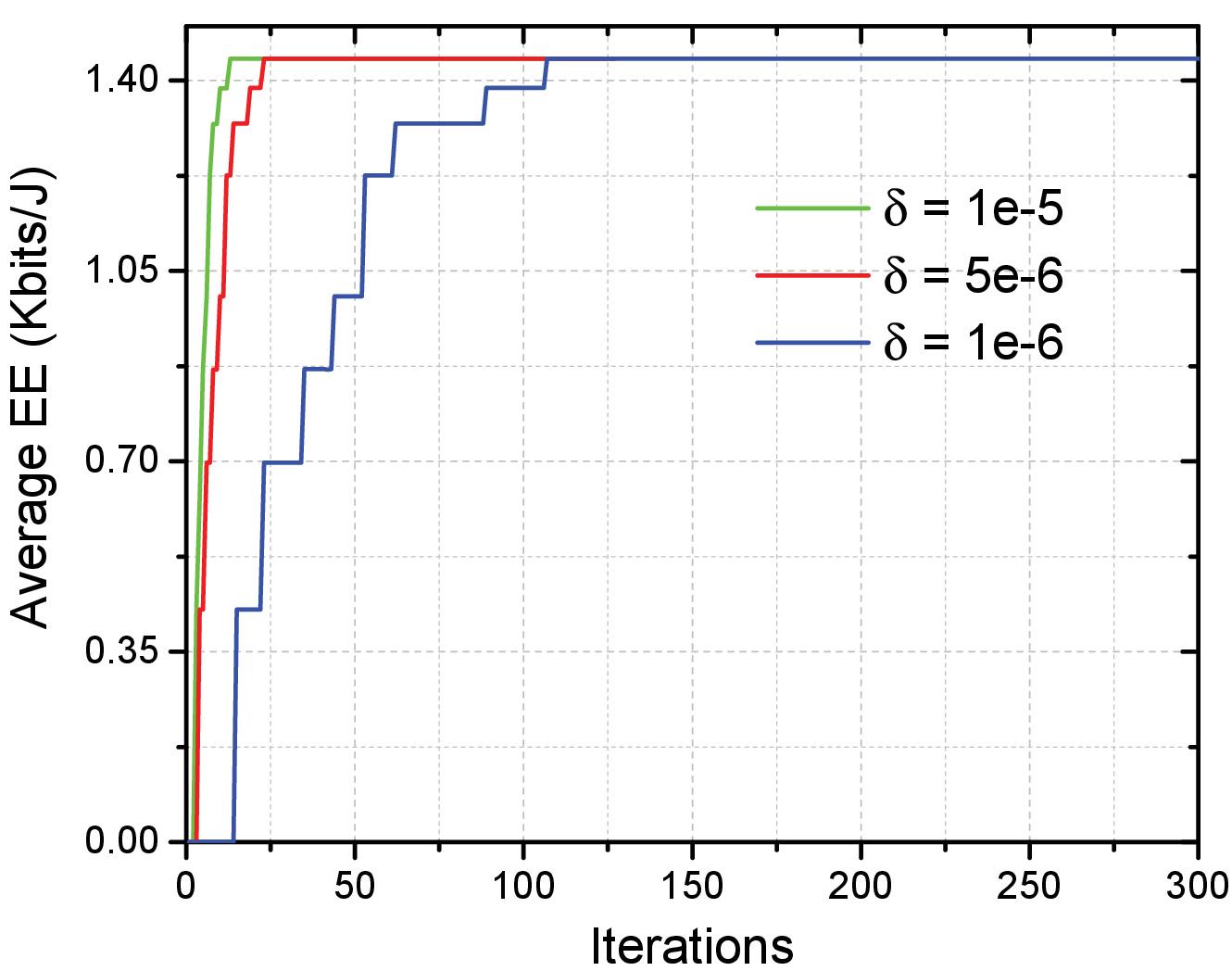}
  \caption{\small Inner loop convergence rate under various step sizes} \label{inner_iter}
\end{subfigure}
\begin{subfigure}[t]{0.25\textwidth}
  \includegraphics[width=\linewidth]{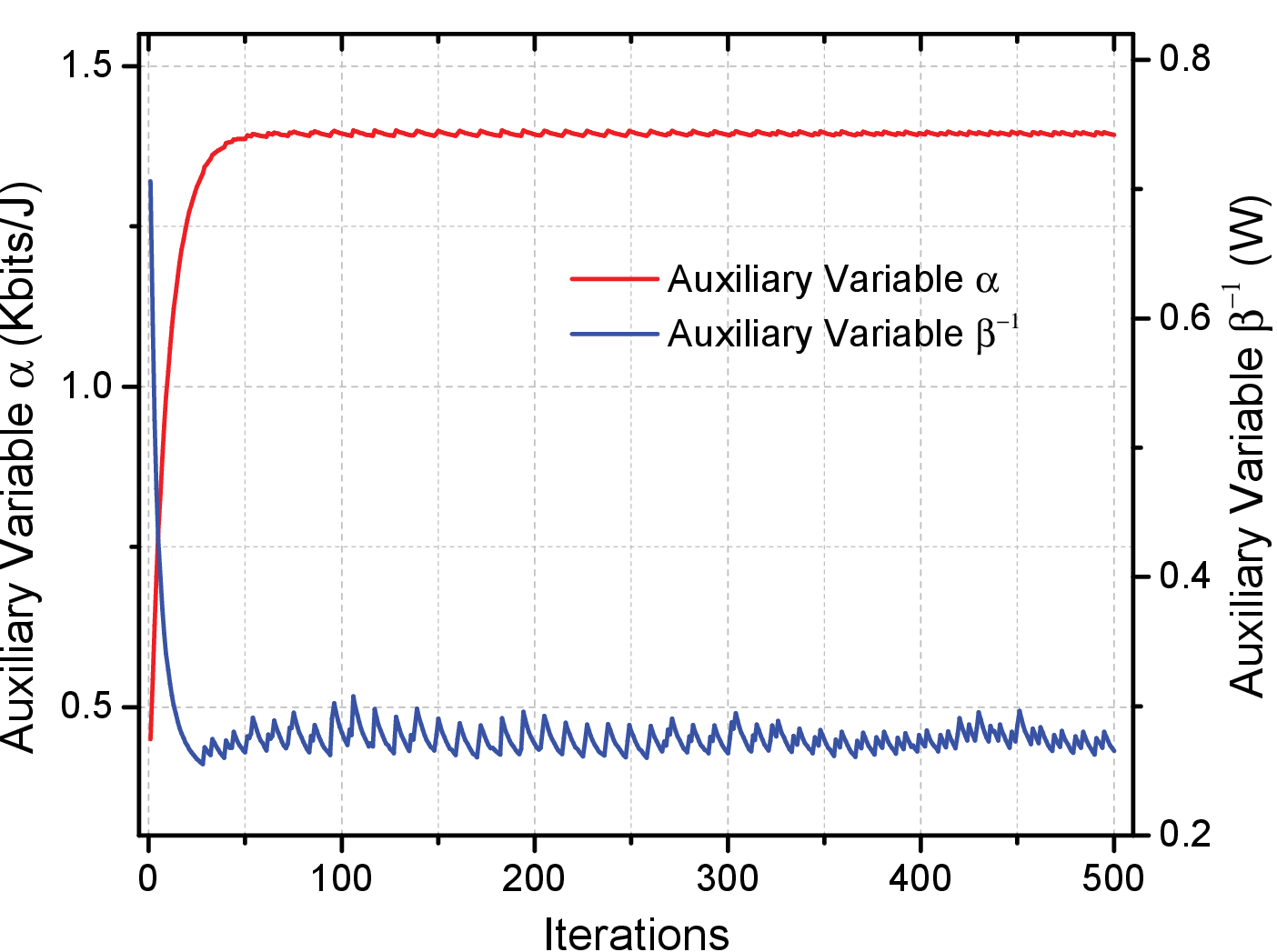}
  \caption{\small Outer loop convergence rate w.r.t. parametric variables} \label{outer_iter}
\end{subfigure}
   \caption{Convergence analysis for Algorithm.\ref{opt_algo}} \label{cov_analysis}
\end{figure}

Given the feasibility of Algorithm \ref{opt_algo}, we now conduct the performance comparison from solving Relaxed-UNEE-Max, Asso-UNEE-Max and Rnd-UNEE-Max utilizing Algorithm \ref{opt_algo} and Algorithm \ref{rnd_algo}, respectively. Besides, by solving (\ref{ben_unee_max}), we obtain the network-wide energy efficiency in the 4G/LTE cellular network, which is used as the benchmark. The evaluation is conducted under different network sizes in terms of the number of end users. We also employ the same values of ${N_{tot}}$, ${P_{u, \max}}$, ${W_m}$ and confidence level $\Delta$ as the previous simulation. The results are shown in Fig.\ref{algo_comp}. It can be seen that these curves demonstrate the same relationship between the network size and the network-wide energy efficiency: as the number of users increase linearly, the network-wide energy efficiency first grows exponentially and then increases slowly. The reason is that the network resources in terms of OFDM sub-channels and harvested band are sufficient when the network size is small and introducing more users will increase the resource utilization efficiency, thus increasing the total network EE. As the number of users keeps increasing, the network becomes congested in the sense that scheduling and routing in the cognitive mesh network becomes the major bottleneck to further boost the network performance. In the later evaluation, we will demonstrate this phenomenon.
\begin{figure}[!htb]
  \begin{center}
  \includegraphics[width=2.8in]{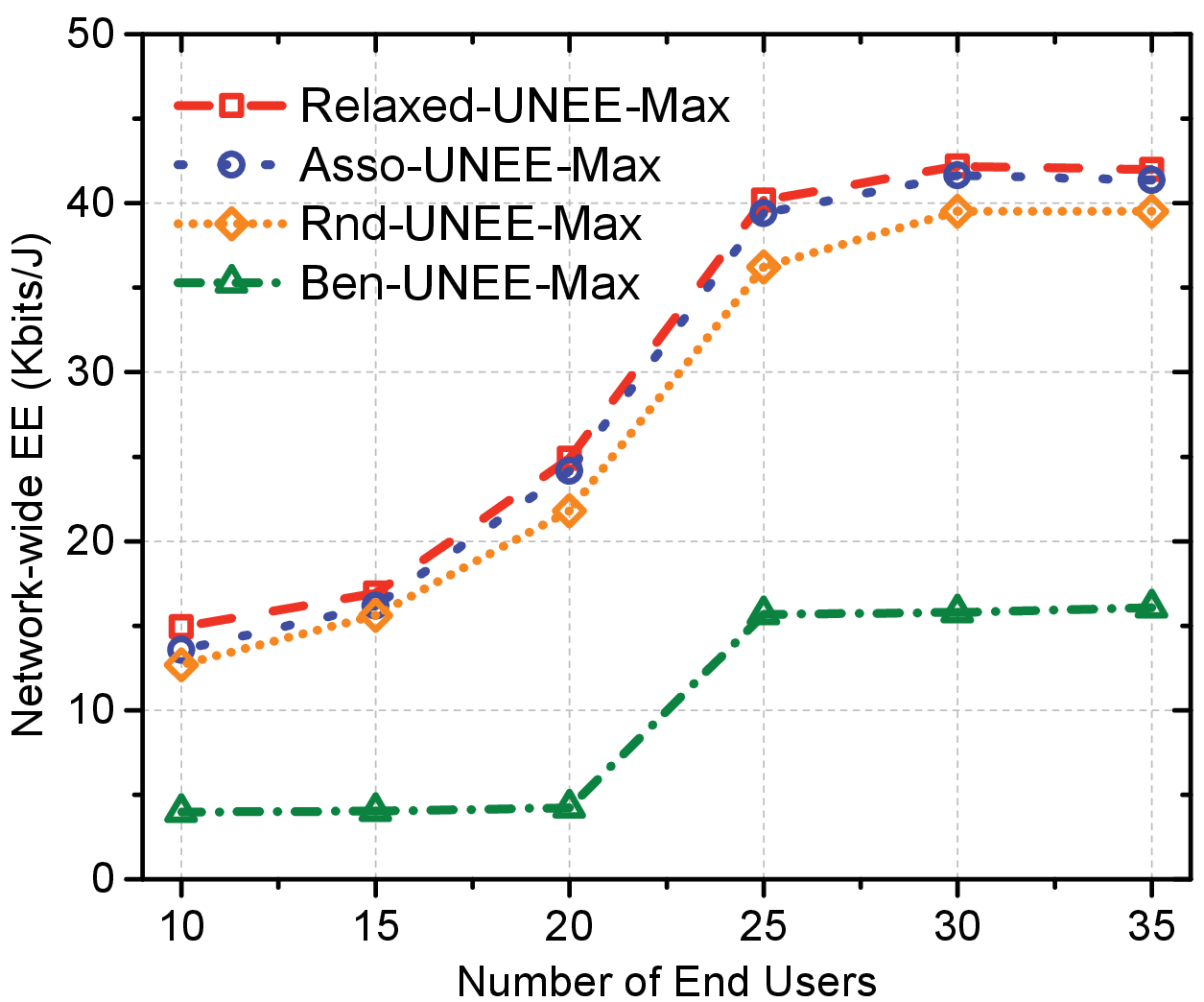}
  \end{center}
  \begin{center}
   \parbox{8cm}{\caption{Performance Comparison for different problems under various network sizes.}\label{algo_comp}}
  \end{center}
\end{figure}

On the other hand, we can see that the solution to the Relaxed-UNEE-Max problem yields the highest network-wide EE since every user can be associated with several infrastructural nodes to take full advantage of network diversity. However, by fixing the association variables and solving the Asso-UNEE-Max problem does not sacrifice too much network performance, as shown in Fig.\ref{algo_comp}. Based on the solution of the Asso-UNEE-Max problem, further applying rounding procedure and solving the Rnd-UNEE-Max problem, gives an even lower network-wide EE. Nevertheless, the optimality gaps between the solutions of Relaxed-UNEE-Max and Rnd-UNEE-Max reduces from 19.35\% to 7.14\% as the number of users increases from 10 to 35, which means our proposed approximation algorithm for association and rounding works well when the network size scales up.

Furthermore, the network-wide energy efficiency in our network is much higher (e.g., 143\% more in the scenario of 25 end users) than that in the traditional cellular network. Such a significant gain in the energy efficiency on one hand attributes to the additional harvested spectrum while on the other hand is due to the close proximity between end users and CR routers allowing lower transmit power for users.

\begin{figure}[!t]
\begin{subfigure}[t]{0.225\textwidth}
  \includegraphics[width=\linewidth]{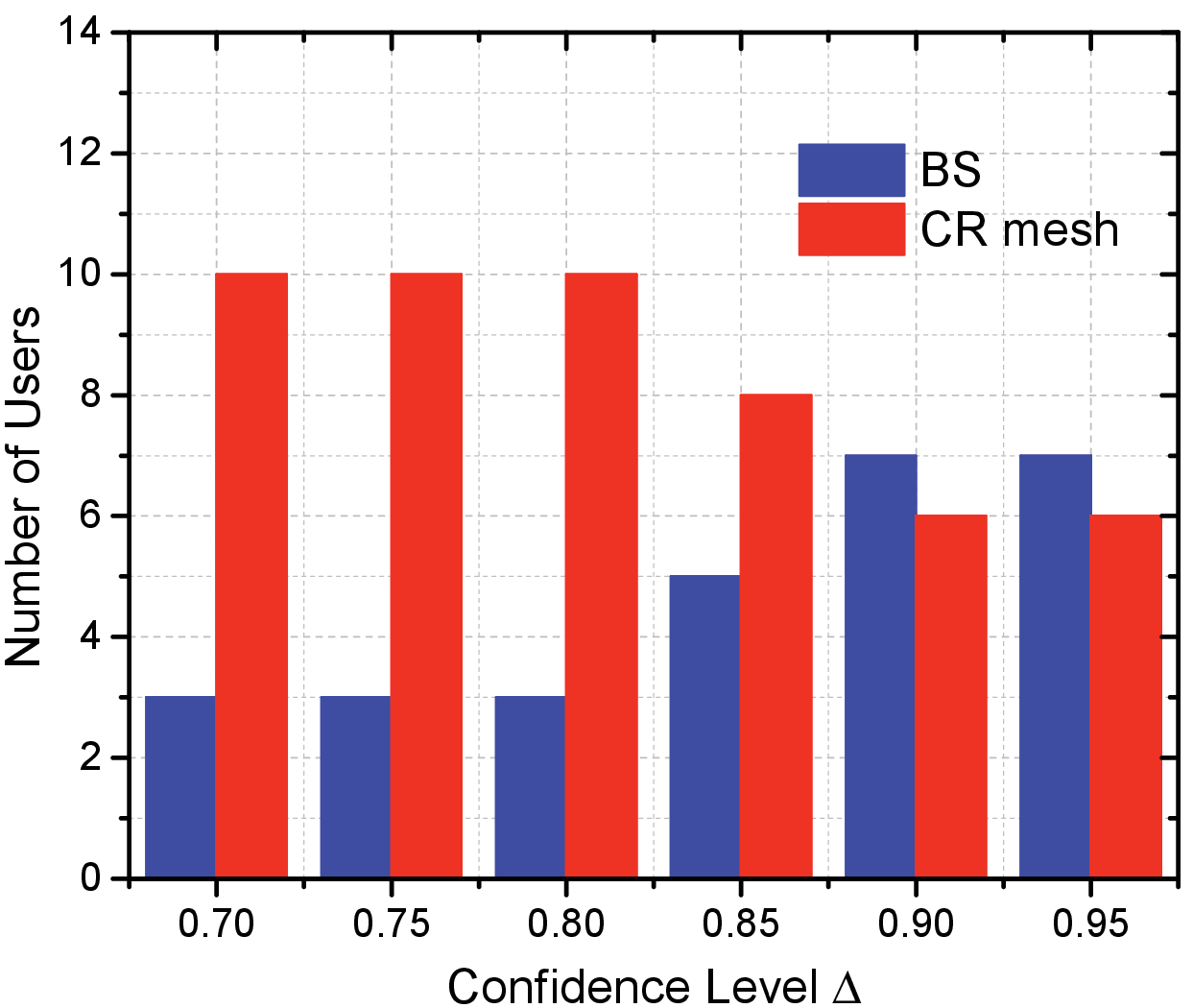}
  \caption{\small Impact of harvested band on user associations} \label{asso_eval}
\end{subfigure}
\begin{subfigure}[t]{0.25\textwidth}
  \includegraphics[width=\linewidth]{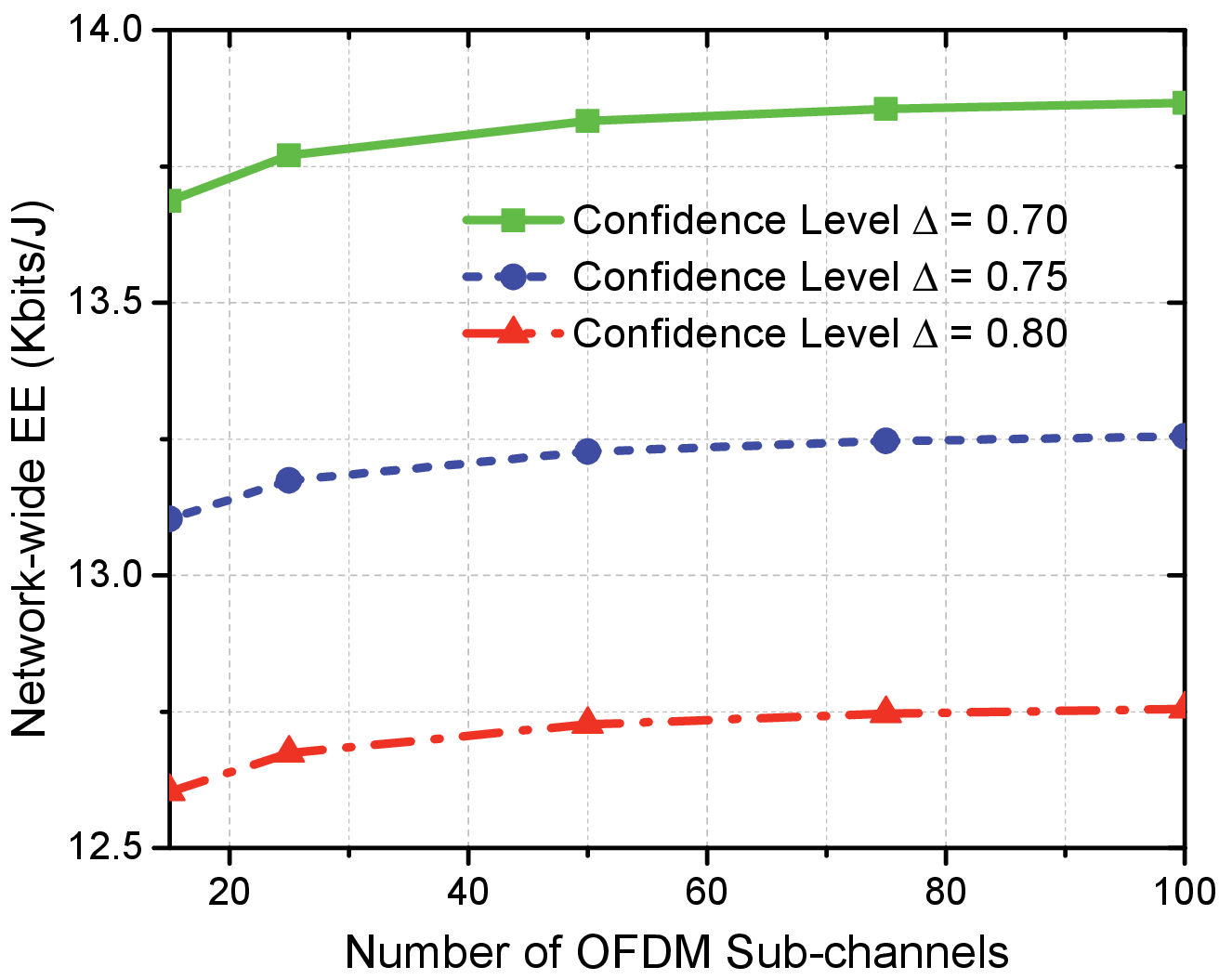}
  \caption{\small Impact of number of OFDM channels on network EE} \label{ofdm_eval}
\end{subfigure}
   \caption{Impact of bandwidth on network performance.} \label{bnd_eval}
\end{figure}
Next, we analyze how the number of OFDM sub-channels and uncertainty of harvested band could affect the user association decision and network performance. For the user association evaluation, we randomly select 13 users just for demonstrative purposes and set ${N_tot}=100$ and ${W_m}=20KHz$, while keeping other parameters the same as before. The result is shown in Fig.\ref{asso_eval}, which illustrates the number of users connected to the BS and to the cognitive mesh network. It can be seen that when the confidence level increases, more users are switched from the cognitive mesh network to the BS. The reason is that higher confidence level means more strict requirement on constraint (\ref{flow_mesh_bnd}), which in other words means that the usable harvested band becomes more limited. Therefore, some users are re-associated with the BS so that their throughput would be higher although they may use higher transmit power.

The harvested band affects the backbone capability, while the number of OFDM sub-channels impacts the capacity of the first hop from end users and infrastructural nodes. As shown in Fig.\ref{ofdm_eval}, we examine how different OFDM sub-channel patterns (e.g., $\{6, 15, 25, 50, 75, 100\}$) influences the network-wide EE. It can be observed that for the fixed uncertainty of the harvested band (i.e., available bandwidth), network-wide EE increases in a decreasing rate as the number of OFDM sub-channels increases. The reason is that as the number of OFDM sub-channels becomes sufficiently large, the available harvested band allocated to the cognitive mesh network becomes the bottleneck to support the traffic on the first hop links. This also explains the observation that the network-wide EE increases as the confidence level decreases for the fixed number of OFDM sub-channels.

\begin{figure}[!htb]
  \begin{center}
  \includegraphics[width=2.8in]{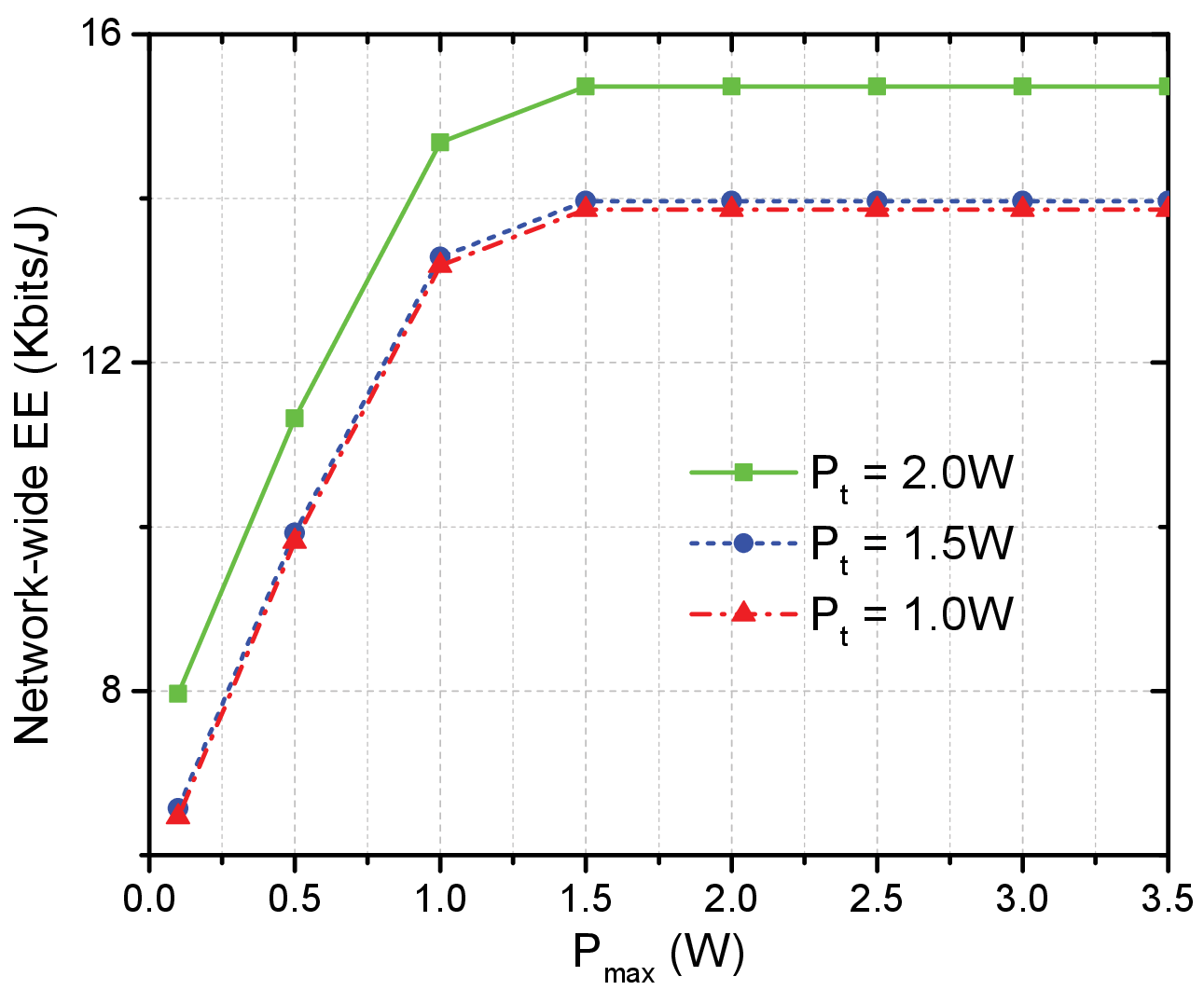}
  \end{center}
  \begin{center}
   \parbox{8cm}{\caption{Impact of users' and CR routers' transmit power on network performance}\label{power_comp}}
  \end{center}
\end{figure}
Another design dimension that could impact the network performance is the transmit power. Here, we examine how the users' transmit power as well as CR routers' transmit power could jointly affect the network-wide EE. The results are shown in Fig.\ref{power_comp}. First of all, we see that the user's maximum allowable power ${P_{max}}$ only affects the network-wide EE at its lower value while the network performance stays constant as ${P_{max}}$ continues to increase. The reason is that end users can utilize very low power for connection and increasing the ${P_{max}}$ would not give a higher transmit power in order to optimize the EE. On the other hand, the relationship between CR routers' transmit power ${P_t}$ and network performance is worth explaining. According to Eq.(\ref{tx_nb}-\ref{int_nb}), ${P_t}$ impacts the communication/interference range, which further influences the construction of the conflict graph. For instance, when ${P_t}=1.0W$, ${R_i^T}=166.7$m and ${R_i^I}=234.5$m; while when ${P_t}=2.0W$, ${R_i^T}=209.9$m and ${R_i^I}=295.46$m. From the network topology shown in Fig.\ref{topology}, we can clearly see that the number of reachable infrastructural nodes for each CR router becomes larger while each CR router's interfered nodes remain the same. As a result, the size of each MIS $q$ increases and more links can be scheduled for transmission at the same time, which enhances the achievable link capacity in the mesh network. Therefore, the network-wide EE increases with the ${P_t}$ increasing from 1W to 2W. It should be noted that this may not always hold true if the power increase incurs more interfered nodes. However, this general trend reflects the fact that by sacrificing the infrastructure's power consumption, end users' EE will be improved, which indicates that the power consumption burden is shifted from light-weighted end devices to the more powerful infrastructural nodes.

\section{Conclusion}
In this work, we investigate the energy efficiency (EE) design of battery-powered devices. Our ideology is to shift their energy consumption to grid-powered devices, thus increasing their EE. This ideology is realized in a cognitive mesh network, in which we model a cross-layer optimization problem to maximize end devices' EE. Specifically, we propose an objective function as the weighted sum of each device's EE and characterize constraints including device association, flow routing, link scheduling, channel allocation and power control. To solve this complex problem, we propose parametric subtractive transformation, $\Delta$-confidence level, critical MISs and integer relax-then-rounding to convert the original problem into a tractable one, and further decouple this large scale optimization problem into a two-layer optimization problem. We conduct extensive simulations to demonstrate the optimality and feasibility of our proposed algorithms and also show how the design variables impact the network performance.

\bibliography{ee}
\bibliographystyle{IEEEtran}

\newpage
\appendices
\section{Proof of Theorem \ref{theory}}
First, we re-write the FCL constraints in problem (\ref{relaxed-unee-max}) using a vector form for compact representation. Given the channel model and transmit power, all the communication links among infrastructural nodes can be calculated and we use $m$ to represent the link of an ordered node pair $(i,j)$ while the set of links is denoted as ${\cal M} = \{ 1,...,m,...,M\}$. For any end user $u$, we can view its flow as a single-commodity model. Now, we introduce an incidence matrix $\bm{{A_1}} \in {\mathbb{Z}^{|{\cal K}| \times M}}$ to represent flow constraint on intermediate infrastructural nodes, where each entry ${a_{i,m}}$ takes value 1 if $i$ is the start node of link $m$; -1 if $i$ is the end node of link $m$; 0 otherwise. We also present $\bm{{A_2}} \in {\mathbb{Z}^{1 \times M}}$ and $\bm{{A_3}} \in {\mathbb{Z}^{1 \times M}}$ to represent flows coming out the BS and flows aggregated at the BS from CR mesh network, respectively, where ${a_{1,m}}$ is selected as 1 or 0. Thus, the FCL constraint can be re-write as follows:
\begin{equation}\label{fcl_vec}
\bm{A} \cdot \bm{f} = \bm{D} \cdot \bm{{\widetilde f_u}}\,, \quad \forall u \in {\cal U}
\end{equation}
where $\bm{A}=[\bm{{A_1}}; \bm{{A_2}}; \bm{{A_3}}] \in {\mathbb{Z}^{|{\cal K}+2| \times M}}$ is the concatenation of three matrices, $\bm{f}={[{f_1},...,{f_m},...,{f_M}]^T}$ represents the links, $\bm{D} \in {\mathbb{Z}^{|{\cal K} + 2| \times |{\cal K}|}}$ is also an incidence matrix consisting of 1s and 0s representing the source of end user's traffic, and $\bm{{\widetilde f_u}} = {[{\widetilde f_{u,1}},...,{\widetilde f_{u,k}},...,{\widetilde f_{u,|{\cal K}|}}]^T}$ represents the user traffic.

We introduce Lagrange multipliers $\bm{\beta}$=$\{ {\beta _1},...,{\beta _u},...,{\beta _U}\}$ associated with EE constraints in (\ref{int-relaxed-unee-max}), $\bm{v}$=$\{ {v_{1,1}},...,{v_{u,1}},...,{v_{U,|{\cal K}|}}\}$ for capacity constraint of the link between end users and infrastructural nodes, $\bm{\psi}$=$\{ {\psi _1},...,{\psi _m},...,{\psi _M}\}$ for the capacity constraint (\ref{flow_mesh_bnd_new}) of the link between infrastructural nodes, $\bm{\varphi}$=$\{ {\varphi _1},...,{\varphi _u},...,{\varphi _U}\}$ for the transmit power constraint, $\bm{\chi}$=$\{ {\chi _{1,1}},...,{\chi _{1,|{\cal K}| + 2}},...,{\chi _{U,|{\cal K}| + 2}}\}$ associated with the FCL constraint (\ref{fcl_vec}) and $\bm{\vartheta}$=$\{ {\vartheta _1},...,{\vartheta _q},...,{\vartheta _{{Q^{'}}}}\}$ for the link scheduling constraint (\ref{lnk_one}). Thus, the Lagrange function of (\ref{int-relaxed-unee-max}) is given by
\begin{equation}
\begin{aligned}
{\cal L}&(\bm{p},\bm{N},\bm{\widetilde f},\bm{f},\bm{\lambda}, \bm{\alpha}; \bm{\beta}, \bm{v}, \bm{\psi}, \bm{\varphi}, \bm{\chi}, \bm{\vartheta}) = \sum\limits_{u \in {\cal U}} {{\omega _u}{\alpha _u}} - \\
&\sum\limits_{u \in {\cal U}} {{\beta _u}[{\alpha _u}(\eta \sum\limits_{k \in {\cal K}} {{p_{u,k}}}  + P) - \sum\limits_{k \in {\cal K}} {{{\widetilde f}_{u,k}}({l_u})} ]} - \\
&\sum\limits_{u \in {\cal U}} {\sum\limits_{k \in {\cal K}} {{v_{u,k}}\left[ {{{\widetilde f}_{u,k}}({l_u}) - {N_{u,k}}W{{\log }_2}(1 + \frac{{{p_{u,k}}|{h_{u,k}}{|^2}}}{{{N_{u,k}}W \cdot {N_0}}})} \right]} } - \\
&\sum\limits_{m \in {\cal M}} {{\psi _m}\left[ {\sum\limits_{{l_u} \in {\cal L}} {{f_m}({l_u})}  - \sum\limits_{q = 1}^{{Q^{'}}} {{\lambda _q}F_{{c_{i,j}}}^{ - 1}(1 - \beta )} } \right]}  - \\
&\sum\limits_{u \in {\cal U}} {{\varphi _u}(\sum\limits_{k \in {\cal K}} {{p_{u,k}}}  - {P_{u,\max }})} - \sum\limits_{u \in {\cal U}} {\sum\limits_{i = 1}^{|{\cal K}| + 2} {{\chi _{u,i}}(\bm{{a_i}f} - \bm{{d_i}{{\widetilde f}_u}})} } \\
&- \sum\limits_{q = 1}^{{Q^{'}}} {{\vartheta _q}({\lambda _q} - 1)},
\end{aligned}
\end{equation}
where the $\bm{{a_i}}=[{a_{i,1}},...,{a_{i,m}},...,{a_{i,M}}]$ is the ${i^{th}}$ row of matrix $\bm{A}$ and $\bm{{d_i}}=[{d_{i,1}},...,{d_{i,k}},...,{d_{i,|{\cal K}|}}]$ is the ${i^{th}}$ row of matrix $\bm{D}$.

Suppose $(\bm{p^{*}}, \bm{N^{*}}, \bm{\widetilde f^{*}}, \bm{f^{*}}, \bm{\lambda^{*}}, \bm{\alpha^{*}})$ are the solutions to problem (\ref{int-relaxed-unee-max}), there exists $\bm{\beta^{*}}, \bm{v^{*}}, \bm{\psi^{*}}, \bm{\varphi^{*}}, \bm{\chi^{*}}, \bm{\vartheta^{*}}$ such that the corresponding KKT conditions are as follows
\begin{equation}
\begin{aligned}
\frac{{\partial {\cal L}}}{{\partial {p_{u,k}}}} &= \beta _u^*\alpha _u^*\eta  - v_{u,k}^{*}\frac{\partial }{{\partial {p_{u,k}}}}(N_{u,k}^*W{\log _2}(1 + \frac{{p_{u,k}^*|{h_{u,k}}{|^2}}}{{N_{u,k}^*W \cdot {N_0}}})) \\
&- \varphi _u^{*} = 0, \qquad \qquad \qquad \qquad \qquad \qquad \forall u, k;
\end{aligned}
\end{equation}
\begin{equation}
\frac{{\partial {\cal L}}}{{\partial {N_{u,k}}}} =  - v_{u,k}^*\frac{\partial }{{\partial {N_{u,k}}}}(N_{u,k}^*W{\log _2}(1 + \frac{{p_{u,k}^*|{h_{u,k}}{|^2}}}{{N_{u,k}^*W{N_0}}})) = 0, \forall u, k;
\end{equation}
\begin{equation}
\frac{{\partial {\cal L}}}{{\partial {{\widetilde f}_{u,k}}}} = \beta _u^* - v_{u,k}^* + \chi _{u,k}^*({d_{k,k}} + \sum\limits_{i = |{\cal K}| + 1}^{|{\cal K}| + 2} {{d_{i,k}}} ) = 0, \forall u, k;
\end{equation}
\begin{equation}
\frac{{\partial {\cal L}}}{{\partial {\lambda _q}}} = \sum\limits_{m \in {\cal M}} {\psi _m^*F_{{c_m}}^{ - 1}(1 - \beta )}  - \vartheta _q^* = 0, \qquad \qquad \qquad  q;
\end{equation}
\begin{equation}\label{kkt_alpha}
\frac{{\partial {\cal L}}}{{\partial {\alpha _u}}} = {\omega _u} - \beta _u^*(\eta \sum\limits_{k \in {\cal K}} {{p_{u,k}}}  + P) = 0 \qquad \qquad \qquad \forall u;
\end{equation}
\begin{equation}\label{kkt_beta}
\beta _u^*\frac{{\partial {\cal L}}}{{\partial {\beta _u}}} = \beta _u^*[\alpha _u^*(\eta \sum\limits_{k \in {\cal K}} {p_{u,k}^*}  + P) - \sum\limits_{k \in {\cal K}} {\widetilde {f_{u,k}^*}} ] = 0, \quad \forall u;
\end{equation}
\begin{equation}
v_{u,k}^*\frac{{\partial {\cal L}}}{{\partial {v_{u,k}}}} = v_{u,k}^*\left[\widetilde {f_{u,k}^*} - N_{u,k}^*W{\log _2}(1 + \frac{{p_{u,k}^*|{h_{u,k}}{|^2}}}{{N_{u,k}^*W {N_0}}})\right] = 0, \forall u, k;
\end{equation}
\begin{equation}
\psi _m^*\frac{{\partial {\cal L}}}{{\partial {\psi _m}}} = \psi _m^*\left[ {\sum\limits_{{l_u} \in {\cal L}} {{f_m}({l_u})}  - \sum\limits_{q = 1}^{{Q^{'}}} {{\lambda _q}F_{{c_m}}^{ - 1}(1 - \beta )} } \right] = 0, \quad \forall m;
\end{equation}
\begin{equation}
\varphi _u^*\frac{{\partial {\cal L}}}{{\partial {\varphi _u}}} = \varphi _u^*(\sum\limits_{k \in {\cal K}} {{p_{u,k}}}  - {P_{u,\max }}) = 0, \qquad \forall u;
\end{equation}
\begin{equation}
\chi _{u,i}^*\frac{{\partial {\cal L}}}{{\partial {\chi _{u,i}}}} = \chi _{u,i}^*({a_{i,k}}f_k^* - {d_{i,k}}\widetilde {f_{u,k}^*}) = 0, \qquad \forall i, u;
\end{equation}
\begin{equation}
\vartheta _q^*\frac{{\partial {\cal L}}}{{\partial {\vartheta _q}}} = \vartheta _q^*({\lambda _q} - 1) = 0, \qquad \forall q;
\end{equation}
Given Eq.(\ref{kkt_alpha}) and (\ref{kkt_beta}) and the power consumption is a non-negative value, we can re-write them in the following forms
\begin{equation}\label{proof_parameter}
\begin{aligned}
&{\alpha _u}^{*} = \frac{{\sum\limits_{k \in {\cal K}} {{{\widetilde f}^*_{u,k}}({l_u})} }}{{\eta \sum\limits_{k \in {\cal K}} {{p_{u,k}}^*} + {P_c}}} \\
&{\beta _u}^{*} = \frac{{{\omega _u}}}{{\eta \sum\limits_{k \in {\cal K}} {{p_{u,k}}^*}  + {P_c}}}
\end{aligned}
\end{equation}
Besides, it is clear that the previous system equations are also the KKT conditions for the problem (\ref{p-relaxed-unee-max}) given the parameters ${\alpha _u}={\alpha _u}^{*}$ and ${\beta _u}={\beta _u}^{*}$. On the other hand, we can follow the similar procedure and prove problem (\ref{p-relaxed-unee-max}) has the identical solution to problem (\ref{int-relaxed-unee-max}) when (\ref{proof_parameter}) holds. Therefore, Theorem.\ref{theory} is proved to be correct.

\ifCLASSOPTIONcaptionsoff
  \newpage
\fi

\end{document}